\documentclass[11pt]{article}

\usepackage{fullpage}
\usepackage{amsmath,amsthm,amssymb}
\usepackage{authblk}
\usepackage{hyperref}
\usepackage{cleveref}
\usepackage{thmtools}
\usepackage{thm-restate}
\usepackage[usenames,dvipsnames,svgnames,table]{xcolor}
\usepackage{graphicx}
\usepackage{caption}
\usepackage{subcaption}

\newtheorem{lemma}{Lemma}
\newtheorem{theorem}{Theorem}
\newtheorem{corollary}{Corollary}
\newtheorem{problem}{Problem}
\newtheorem{myclaim}{Claim}
\newtheorem{subclaim}{Subclaim}

\theoremstyle{definition}
\newtheorem{definition}[theorem]{Definition}
\newtheorem{remark}{Remark}

\newenvironment{proofof}[1]{\medskip\noindent\emph{Proof of #1. }\ignorespaces}{\hfill\qed\medskip\par\noindent\ignorespacesafterend}
\newcommand{\qedclaim}{\hfill $\diamond$ \medskip}
\newenvironment{proofclaim}{\noindent\ignorespaces{\em Proof.}}{\hfill\qedclaim\par\noindent} 
\newcommand{\qedsubclaim}{\hfill $\circ$ \medskip}
\newenvironment{proofsubclaim}{\noindent{\em Proof.}}{\qedsubclaim}

\newsavebox{\mybox}

\begin{document}

\title{Diameter computation on $H$-minor free graphs and graphs of bounded (distance) VC-dimension}

\author[1]{Guillaume Ducoffe}
\author[2]{Michel Habib}
\author[3]{Laurent Viennot}

\affil[1]{University of Bucharest, Faculty of Mathematics and Computer Science, and National Institute for Research and Development in Informatics, Romania\footnote{This work was supported by project PN 19 37 
04 01 ``New solutions for complex problems in current ICT research fields based on modelling and optimization'', funded by the Romanian Core Program of the Ministry of Research and Innovation (MCI) 2019-2022. This work was also supported by a grant of Romanian Ministry of Research and Innovation CCCDI-UEFISCDI. project no. 17PCCDI/2018.}}
\affil[2]{Paris University, and IRIF CNRS, France\footnote{Supported by Inria Gang project-team, and ANR project DISTANCIA (ANR-17-CE40-0015).}}
\affil[3]{Inria, and Paris University, France\footnote{Supported by Irif laboratory from CNRS and Paris University, and ANR project Multimod (ANR-17-CE22-0016).}}

\date{}

\maketitle

\begin{abstract}
Under the Strong Exponential-Time Hypothesis, the diameter of general unweighted graphs cannot be computed in truly subquadratic time.
Nevertheless there are several graph classes for which this can be done such as bounded-treewidth graphs, interval graphs and planar graphs, to name a few.
We propose to study unweighted graphs of constant {\em distance VC-dimension} as a broad generalization of many such classes
-- where the distance VC-dimension of a graph $G$ is defined as the VC-dimension of its ball hypergraph: whose hyperedges are the balls of all possible radii and centers in $G$.
In particular for any fixed $H$, the class of $H$-minor free graphs has distance VC-dimension at most $|V(H)|-1$.
\begin{itemize}
\item Our first main result is a Monte Carlo algorithm that on graphs of distance VC-dimension at most $d$, for any fixed $k$, either computes the diameter or concludes that it is larger than $k$ in time $\tilde{\cal O}(k\cdot mn^{1-\varepsilon_d})$, where $\varepsilon_d \in (0;1)$ only depends on $d$. 
We thus obtain a {\em truly subquadratic-time parameterized} algorithm for computing the diameter on such graphs.

\item Then as a byproduct of our approach, we get the first truly subquadratic-time randomized algorithm for {\em constant} diameter computation on all the {\em nowhere dense} graph classes. The latter classes include all proper minor-closed graph classes, bounded-degree graphs and graphs of bounded expansion.

\item Finally, we show how to remove the dependency on $k$ for {\em any} graph class that excludes a fixed graph $H$ as a minor.
More generally, our techniques apply to any graph with constant distance VC-dimension and {\em polynomial expansion} (or equivalently having strongly sublinear balanced separators).
As a result for all such graphs one obtains a truly subquadratic-time randomized algorithm for computing their diameter.
\end{itemize}
We note that all our results also hold for {\em radius} computation.
Our approach is based on the work of Chazelle and Welzl who proved the existence of spanning paths with strongly sublinear {\em stabbing number} for every hypergraph of constant VC-dimension.
We show how to compute such paths efficiently by combining known algorithms for the stabbing number problem with a clever use of $\varepsilon$-nets, region decomposition and other partition techniques.
\end{abstract}

\section{Introduction}\label{sec:intro}

In this paper we present new results on exact diameter computation within several classes of unweighted (undirected) graphs with a geometric flavor.
We recall that the diameter of an unweighted graph is the maximum number of edges on a shortest path.
Beyond its many practical applications, this fundamental problem in Graph Theory has attracted a lot of attention in the fine-grained complexity study of polynomial-time solvable problems~\cite{AVW16,BRSV+18,BCH16,CGR16,CLRG+14,Dah6,Duc19,EvD16,RoV13}.
More precisely, for every $n$-vertex $m$-edge unweighted graph the textbook algorithm for computing its diameter runs in time ${\cal O}(nm)$.
In a seminal paper~\cite{RoV13} this roughly quadratic running-time (in the size $n+m$ of the input) was matched by a quadratic {\em lower-bound}, assuming the Strong Exponential-Time Hypothesis (SETH).
We stress that for graphs with millions of nodes and edges, quadratic time is already prohibitive.

The conditional lower-bound of~\cite{RoV13} also holds for sparse graphs {\it i.e.}, with only $m = {\cal O}(n)$ edges~\cite{AVW16}.
However it does {\em not} hold for many well-structured graph classes~\cite{AVW16,BCD98,BHM18,CDHP01,Cab18,CDP18,Dam16,Duc19,Epp00,FaP80,GKHM+18,Ola90}. 
Our work proposes some new advances on the characterization of graph families for which we can compute the diameter in truly subquadratic time.

\subsection{Related work}

Before we detail our contributions, we wish to mention a few recent (and not so recent) results that are most related to our approach.

\paragraph{Interval graphs.}
An early example of linear-time solvable special case for diameter computation is the class of interval graphs~\cite{Ola90}.
For every interval graph $G$ and for any integer $k$, if we first compute an interval representation for $G$ in linear-time~\cite{HMPV00} then we can compute by dynamic programming, for every vertex $v$, the contiguous segment of all the vertices at a distance $\leq k$ from $v$ in $G$.
It takes almost linear-time and it implies a straightforward quasi linear-time algorithm for diameter computation.
More efficient algorithms for diameter computation on interval graphs and related graph classes were proposed in~\cite{CDHP01}.
Nevertheless we will show in what follows that interval orderings are a powerful tool for diameter computation on more general geometric graph classes. 

\paragraph{Bounded-treewidth graphs.}
More recently, quasi linear-time algorithms for diameter computation on bounded-treewidth graphs were presented in~\cite{AVW16,BHM18} with almost optimal dependency on the treewidth parameter.
The cornerstone of these algorithms is the use of $k$-range trees in order to detect the furthest pairs that are disconnected by some small-cardinality separators.
Since then a few other applications of $k$-range trees and, more generally, orthogonal range searching for diameter computation, have been presented in~\cite{Duc19,DHV19+}.
In our work we uncover deeper connections between diameter computation and range searching techniques from computational geometry.

\paragraph{Planar graphs.}
Finally, in a recent breakthrough paper~\cite{Cab18}, Cabello presented the first truly subquadratic algorithm for diameter computation on planar graphs (see also~\cite{GKHM+18} for improvements on his work).
For that he combined $r$-divisions: a recursive decomposition technique for planar graphs and other hereditary graph classes with sublinear balanced separators, with a clever use of additively weighted Voronoi diagrams.
Cabello conjectured that his algorithm could be generalized to bounded-genus graphs. The long version of \cite{GKHM+18} indicates that their techniques could allow such a generalization if computing the diameter of a graph embedded onto a surface of genus $g$ reduces to the planar case with $O(g)$ holes in the regions of some $r$-division. Although it is known that such a graph can be decomposed into planar subgraphs by removing $2g$ shortest paths~\cite{KKS11,Eppstein03}, such reduction is not clear, and we could not find references formally supporting this.
More recently, Li and Parter proposed a distributed algorithm for planar diameter which is based on metric compression~\cite{LiP19} and uses a VC-dimension argument to bound the number distance profiles with respect to a given subset of nodes.
Following the basics of planar algorithms, we partly reuse $r$-divisions within our algorithms.
However we replace the intricate use of Voronoi diagrams with a quite different approach that is based on some interval representations of the balls of a given radius in a graph. Our approach is also based on a VC-dimension argument but in a very different way than~\cite{LiP19}.
In doing so, we can obtain truly subquadratic-time algorithms for diameter computation on bounded genus graphs (and more generally, on any proper minor-closed graph family) while avoiding a great deal of topological complications. Note that our approach works similarly for computing the radius whereas it not clear whether the Voronoi diagram approach does.

\medskip
We stress that for the three aforementioned graph classes, the techniques used for computing their diameter are quite different from each other.
Our work is a first step toward unifying all these previous results for unweighted graphs in a single framework (note that some of the aforementioned results also hold in the directed weighted case).

\subsection{Our contributions}

We study the parameterization of graph diameter by the {\em VC-dimension} of various hypergraphs.
More precisely, a set $Y$ is {\em shattered} by a hypergraph ${\cal H}$ if by intersecting $Y$ with all hyperedges of ${\cal H}$ one obtains the power-set of $Y$.
The VC-dimension of ${\cal H}$ is then defined as the largest cardinality of a subset shattered by ${\cal H}$.
This powerful notion was first introduced by Vapnik and Chervonenkis in~\cite{VaC15}.
Since then it has found applications in sampling complexity and machine learning, among other domains.
We refer to~\cite{KKR+97} for early work on VC-dimension in graphs.
In particular, the VC-dimension of a graph $G$ is defined as the VC-dimension of its closed neighbourhood hypergraph: whose hyperedges are the closed neighbourhoods of vertices in $G$.
Graphs of bounded interval number and proper minor-closed graph classes are two examples of graph families with a {\em constant} upper-bound on their VC-dimension~\cite{KKR+97,CEV07}.

\paragraph{First example.}
As an appetizer we first consider an $n$-vertex split graph with clique-number $\log^{{\cal O}(1)}{n}$, that is a notouriously hard case for diameter computation~\cite{BCH16}.
Given such a split graph $G$ with stable set $S$ and maximal clique $K$, we can pre-process $G$ in linear-time so as to partition the vertices of $S$ into {\em twin classes}: with two vertices in $S$ being called twins if and only if they have the same neighbourhood in $K$ ({\it e.g.}, see~\cite{CDP18}).
If the VC-dimension of $G$ is at most $d$ then, by the Sauer-Shelah-Perles Lemma~\cite{Sau72,She72} the number of twin classes is in ${\cal O}(|K|^d) = \log^{{\cal O}(d)}{n}$.
Therefore, after some linear-time preprocessing, we are left with computing the diameter on a graph of {\em polylogarithmic order}!
Unfortunately, such simple brute-force arguments are no longer sufficient for split graphs of arbitrary clique-size.

\paragraph{Overview of our techniques.}
In order to generalize our approach to any graph of constant VC-dimension, we use the central notion of {\em spanning paths with low stabbing number}. 
Chazelle and Welzl~\cite{ChW89} defined a spanning path for a hypergraph ${\cal H}$ as a total ordering of its vertex-set.
The {\em stabbing number} of such a path is, up to $1$, the maximum number of maximal intervals of which a hyperedge in ${\cal H}$ can be the union (we refer to Sec.~\ref{sec:def} for a formal definition).

Assume for now that we are given a spanning path with stabbing number $t$ for the closed neighbourhood hypergraph of $G$.
Then in linear time, we can compute for every vertex $v$ the ends of the ${\cal O}(t)$ intervals of which $N_G[v]$ is the union.
We denote this set of intervals by $I(v)$ in what follows.
Then, in order to decide whether $G$ has diameter at most two, it is sufficient to check whether for every vertex $u$ we have $\bigcup_{v \in N_G[u]} I(v) = V$.
Since we only need to consider the ends of such intervals, this verification phase takes time ${\cal O}(deg_G(u) \cdot t)$ for a vertex of degree $deg_G(u)$, and so, ${\cal O}(tm)$ total time.
Note that such running-time is always subquadratic if $t$ is sublinear in $n$.
Overall, we reduced the diameter-two problem to the computation of a spanning path with low stabbing number for the closed neighbourhood hypergraph.

Motivated by range searching problems, Chazelle and Welzl proved the existence of spanning paths with {\em strongly sublinear} stabbing number for every hypergraph of constant VC-dimension~\cite{ChW89}!
Following this approach, we obtain our first main result in this paper:

\begin{restatable}{theorem}{diamTwo}
\label{thm:diamTwo}
For every $d > 0$, there exists a constant $\varepsilon_d \in (0;1)$ such that in deterministic time $\tilde{\cal O}(mn^{1-\varepsilon_d})$ we can decide whether a graph of VC-dimension at most $d$ has diameter two.
\end{restatable}

We stress that in contrast to Theorem~\ref{thm:diamTwo}, under the Strong Exponential-Time we cannot decide whether a {\em general} graph has diameter at most two in truly subquadratic time~\cite{RoV13}. 

\smallskip
On our way to prove Theorem~\ref{thm:diamTwo} our main difficulty was to show how to compute for a hypergraph ${\cal H}$ a spanning path of low stabbing number.
Computing a spanning path of minimum stabbing number is NP-hard~\cite{BGRS04}.
However, there exist approximation algorithms for this problem that run in polynomial time~\cite{BGRS04,Har09}.
Their approximation ratio is logarithmic, that is fine for our applications.
Unfortunately, the fastest known algorithms require us to solve a linear program.
So far, the best known algorithms for this intermediate problem run in superquadratic time~\cite{CLS19}.
We show how to decrease the running-time of this part, at the price of a slightly increased stabbing number. For that, we carefully apply the deterministic algorithm resulting from~\cite{ChW89} to some arbitrary partition of ${\cal H}$ in subhypergraphs of sublinear size. 
This nice trick might be of independent interest. We thus state the following theorem where the size of a hypergraph is defined as the sum of its hyperedge cardinalities.

\begin{restatable}{theorem}{stabbingNb}
\label{thm:stabbingNb}
For every $d > 0$, there exists a constant $\varepsilon_d \in (0;1)$ such that in $\tilde{\cal O}(m + n^{2-\varepsilon_d})$ deterministic time, for every $n$-vertex hypergraph ${\cal H}$ of VC-dimension at most $d$ and size $m$, we can compute a spanning path of stabbing number $\tilde{\cal O}(n^{1-\varepsilon_d})$.

Moreover, $\varepsilon_d = \frac 1 {2^{d+1}[ c(d+1) - 1 ] + 1}$ for some universal constant $c > 2$.
\end{restatable}

\paragraph{From VC-dimension to {\em distance} VC-dimension.}
In order to go beyond Theorem~\ref{thm:diamTwo}, we need to consider a stronger notion of VC-dimension for graphs.
The {\em distance VC-dimension}\footnote{Our definition of distance VC-dimension is slightly weaker than the one proposed in~\cite{BoT15}.} of $G$ is equal to the VC-dimension of its {\em ball hypergraph}: of which the hyperedges are all possible balls in $G$.
Note that a bounded distance VC-dimension implies a bounded VC-dimension, but the converse a priori does not hold.
Nevertheless, and perhaps surprisingly, there are still many classes of graphs with a constant distance VC-dimension.
These classes include, among others: interval graphs, planar graphs~\cite{CEV07} and, more generally, any proper minor-closed graph family(from Remark~3 in~\cite{CEV07}), as well as graphs of bounded rank-width~\cite{BoT15}.

\begin{restatable}{theorem}{boundedDiam}
\label{thm:boundedDiam}
There exists a Monte Carlo algorithm such that, for every positive integers $d$ and $k$, we can decide whether a graph of distance VC-dimension at most $d$ has diameter at most $k$. The running time is in $\tilde{\cal O}(k \cdot mn^{1-\varepsilon_d})$, where $\varepsilon_d \in (0;1)$ only depends on $d$.
\end{restatable}

Eppstein proved in~\cite{Epp00} that for any {\em constant} $k$, we can decide in linear time whether the diameter of a planar graph is at most $k$.
Our result can be seen as a generalization of his to any graph class of constant distance VC-dimension -- but at the price of a superlinear running-time.
Furthermore, our techniques also apply to superconstant diameters, say polylogarithmic in $n$, or even polynomial in $n$ provided the exponent is in $o(\varepsilon_d)$.

\smallskip
Our main technical contribution in this part is the efficient computation of spanning paths with strongly sublinear stabbing number for some {\em dense} hypergraphs of constant VC-dimension. 
More precisely, the {\em $\ell$-neighbourhood hypergraph} of $G$ has for hyperedges the balls of radius $\ell$ in $G$.  
For instance, the $1$-neighbourhood hypergraph of $G$ is exactly its closed neighbourhood hypergraph.
In order to prove Theorem~\ref{thm:boundedDiam}, we reduce the problem of deciding whether a graph has diameter at most $k$ to the computation of a spanning path with low stabbing number for its $(k-1)$-neighbourhood hypergraph.
In this sense, the proofs of Theorems~\ref{thm:diamTwo} and~\ref{thm:boundedDiam} are very similar.
However, an additional difficulty here is that we cannot have direct access to this $(k-1)$-neighbourhood hypergraph.
Indeed, in the worst case all hyperedges of this hypergraph may have a cardinality in $\Omega(n)$, and then storing the hypergraph itself would already require quadratic space.

We overcome this issue by computing an {\em $\varepsilon$-net}~\cite{HaW87,VaC15} in order to partition the vertices of the graph in a small number of groups, with every two vertices in the same group having almost the same ball of radius $k-1$.
By selecting only one vertex per group, we so reduce the number of hyperedges ({\it i.e.}, balls of radius $k-1$) to be considered. 
Finally, once a spanning path was computed for this smaller hypergraph, for every unselected vertex we compute the symmetric difference between its ball of radius $k-1$ and the one of the unique vertex taken in its group.
Our solution in order to do that efficiently is to first compute a spanning path with low stabbing number for the {\em $(k-2)$-neighbourhood hypergraph}.
This is where the dependency on $k$ occurs, as overall we will need to compute a spanning path for $k-1$ consecutive hypergraphs.
Our algorithm is randomized and succeeds with high probability. The use of randomization comes from the $\varepsilon$-net construction. Although deterministic algorithms do exist for that~\cite{BCM99}, it is not clear whether they can be used as efficiently as the simple sampling technique of the randomized algorithm. We leave open the question of finding a deterministic variant of Theorem~\ref{thm:boundedDiam}.

\medskip
We note that this above technique can be applied under slightly weaker hypothesis than the one we state in Theorem~\ref{thm:boundedDiam}.
For instance, Ne\v{s}et\v{r}il and Ossona de Mendez proved that for all {\em nowhere dense} graph classes (i.e., a broad generalization of proper minor-closed graph classes and bounded-degree graphs), for any graph in the class and for any constant $k$, the VC-dimension of the $k$-neighbourhood hypergraph is constantly upper-bounded~\cite{NeO16}.
It allows us to derive the following weaker version of our Theorem~\ref{thm:boundedDiam}:

\begin{restatable}{theorem}{nowhereDense}
\label{thm:nowhereDense}
Let ${\cal G}$ be a class of nowhere dense graphs.
There exists a Monte Carlo algorithm such that, for every constant $k = {\cal O}(1)$, for any graph in ${\cal G}$ we can decide whether its diameter is at most $k$ in time $\tilde{\cal O}(mn^{1-\varepsilon_{\cal G}(k)})$, for some constant $\varepsilon_{\cal G}(k) \in (0;1)$ that only depends on $k$. 
\end{restatable} 

Let us mention that under SETH, Theorem~\ref{thm:nowhereDense} is the best result that we can hope for nowhere dense graph classes.
Indeed, bounded-degree graphs are nowhere dense and, under SETH, we cannot compute their diameter in truly subquadratic time even if it is in $\omega(\log{n})$~\cite{EvD16}.

\medskip
We conjecture that on every graph family of constant distance VC-dimension, we can compute the diameter in truly subquadratic time.
Our next main result shows the conjecture to be true for any monotone graph family with strongly sublinear balanced separators, {\it a.k.a} the graphs of {\em polynomial expansion}~\cite{DvN16}.

\begin{restatable}{theorem}{subLin}
\label{thm:subLin}
Let ${\cal G}$ be a monotone graph class with strongly sublinear balanced separators.
Then there exists a Monte Carlo algorithm such that, for every $d > 0$, we can compute the diameter of any graph in ${\cal G}$ of distance VC-dimension at most $d$ in time $\tilde{\cal O}(n^{2-\varepsilon_{\cal G}(d)})$, for some constant $\varepsilon_{\cal G}(d) \in (0;1)$ that only depends on $d$.
\end{restatable} 

Let us recall that $H$-minor free graphs have a constant distance VC-dimension from Remark~3 in~\cite{CEV07} (see also \cite{BoT15}), and that they all have strongly sublinear balanced separators~\cite{AST90,KaR10,Wul11}. Therefore, as an important consequence of Theorem~\ref{thm:subLin}, we get a truly subquadratic-time algorithm for computing the diameter on all the proper minor-closed graph classes.

\smallskip
It might be tempting, in the above Theorem~\ref{thm:subLin}, to drop the assumption that the distance VC-dimension must be bounded. 
Unfortunately, this cannot be done assuming SETH.
Indeed, there is also an equivalence between the graphs of strongly sublinear treewidth and those monotone graph classes with strongly sublinear balanced separators~\cite{DvN19}; however it follows from~\cite{AVW16} that under SETH, we cannot compute the diameter in truly subquadratic time already for $n$-vertex graphs of treewidth $\omega(\log{n})$.
Conversely, not all graph classes with constant distance VC-dimension have strongly sublinear separators.  
This can be seen, {\it e.g.}, with interval graphs.

\medskip
The speed-up of Theorem~\ref{thm:subLin} follows from a faster computation of spanning paths for the neighbourhood hypergraphs.
More precisely, we explain how to compute a spanning path for the $2k$-neighbourhood hypergraph of $G$ from a spanning path of its $k$-neighbourhood hypergraph.
Note that in doing so, we only need to consider {\em logarithmically} many intermediate hypergraphs in order to compute such spanning path.
Our approach for that consists in computing a first (suboptimal) representation of the $2k$-neighbourhood of every vertex.
Then, as for Theorem~\ref{thm:boundedDiam}, we partition the vertices into a small number of groups and we select a unique vertex in each group.
The suboptimal representations are used at the end of the algorithm in order to compute, for every unselected vertex, the symmetric difference between its ball of radius $2k$ and the one of the unique vertex taken in its group. 
So the problem becomes how to compute efficiently these suboptimal representations?

\smallskip
For that, we use a rather classical divide-and-conquer approach.
Federickson~\cite{Fed87} proved that a planar graph can be edge-covered with ${\cal O}(n/r)$ subgraphs of order at most $r$ such that at most ${\cal O}(\sqrt{r})$ vertices of each subgraph can be contained in another subgraph of this decomposition.
His construction directly follows from the planar separator theorem of Lipton and Tarjan~\cite{LiT79}, and as such it can be easily adapted for any monotone graph family with sublinear balanced separators~\cite{HKRS97}\footnote{Note that Federickson proposed several refinements of his construction in~\cite{Fed87}, some of which do use the fact that the input graph is planar. We will use in our proofs an even weaker version of his result than the one presented in this introduction.}.
For illustrating our method, we now focus in this introduction on the planar case.
We can first compute, for some well-chosen $r = n^{\gamma}, \ \gamma \in (0;1)$, a decomposition as described above.
For every two vertices in a same subgraph, we can check whether they are at distance at most $2k$ by checking whether their balls of radius $k$ intersect; assuming $r$ is small enough, and we precomputed a spanning path with low stabbing number for the $k$-neighbourhood hypergraph, this phase can be implemented in order to run in truly subquadratic time.
Then for every subgraph of the decomposition, we compute a breadth-first search from each of the ${\cal O}(\sqrt{r})$ boundary vertices that are also contained in another subgraph.
Overall, there can only be ${\cal O}(n/\sqrt{r})$ such boundary vertices, and so, it takes truly subquadratic time.
Furthermore in doing so, we computed for every subgraph of the decomposition the ${\cal O}(r\sqrt{r})$ distances between the boundary vertices and all the others. 
For any vertex $v$ that is {\em not} on the boundary, we observe that a vertex in another subgraph can be at a distance $\leq 2k$ from $v$ if and only if it is at distance $\leq 2k - dist_G(v,x)$ from some vertex $x$ on the boundary (${\cal O}(\sqrt{r})$ balls to be considered).
Our strategy consists in computing a spanning path with low stabbing number for some ``boundary hypergraph'' whose hyperedges are the ${\cal O}(r\sqrt{r} \times (n/r)) = {\cal O}(n\sqrt{r})$ balls that we need to consider.
We encounter a similar problem as for Theorem~\ref{thm:boundedDiam} because storing this hypergraph may require superquadratic space.
Fortunately, we can encode this hypergraph in a much more compact way by taking advantage of \texttt{(i)} the fact that we can only have ${\cal O}(n/\sqrt{r})$ different centers for the balls, and \texttt{(ii)} that all the balls with a same center have a chain-like inclusion structure.

\medskip
Although we keep the focus on computing the diameter, we shall stress in Sec.~\ref{sec:stabbing-num} that all our techniques can also be applied to {\em radius} computation (i.e., see Remark~\ref{rk:radius}).
Our algorithms almost need no particular information about the graph structure in order to be applied.
In fact, we do not even need to compute the (distance) VC-dimension of the input graph!
From the applicative point of view, this observation (further discussed in Sec.~\ref{sec:stabbing-num}) is quite important.
Indeed, computing the VC-dimension is W[1]-hard~\cite{DEF93} and LogNP-hard~\cite{PaY96}.

\subsection{Organization of the paper}

In Sec.~\ref{sec:def} we formally introduce the concepts of (distance) VC-dimension and stabbing number, along with some of their basic properties.
Then, we explain in Sec.~\ref{sec:spanning-path} how to compute a spanning path with strongly sublinear stabbing number for a hypergraph of constant VC-dimension (Theorem~\ref{thm:stabbingNb}).
As a direct application, we give a short proof of Theorem~\ref{thm:diamTwo}.
Our techniques are generalized in Sec.~\ref{sec:diam-comput} so as to prove Theorems~\ref{thm:boundedDiam} and~\ref{thm:nowhereDense}.
Finally, our main technical result (Theorem~\ref{thm:subLin}) is proved in Sec.~\ref{sec:refinements}.
For that, we will need to recall some useful results on the graphs of {\em polynomial expansion}~\cite{DvN16}.
We discuss some possible future work in Sec.~\ref{sec:ccl}.

\section{Preliminaries}\label{sec:def}

After recalling a few basic definitions about graphs and hypergraphs (Sec.~\ref{sec:graph} and~\ref{sec:hypgraph}) we introduce our framework for computing the diameter of a graph in Sec.~\ref{sec:vc-dim} and~\ref{sec:stabbing-num}. 

\subsection{Graphs and Diameter}\label{sec:graph}
For any undefined graph terminology, see~\cite{BoM08}.
Throughout all this paper we only consider graphs that are undirected, unweighted and connected.
For every graph $G=(V,E)$, let $n := |V|$ be its order and $m := |E|$ be its size.
We denote by $N_G(v)$ and $N_G[v] := N_G(v) \cup \{v\}$ the open and closed neighbourhoods of vertex $v$, respectively.
The degree of $v$ is equal to $|N_G(v)|$ and is denoted by $deg_G(v)$ in what follows.
The length of a path is its number of edges, and the distance $dist_G(u,v)$ between $u,v \in V$  is equal to the length of a shortest $uv$-path.
For every $v \in V$ and $k \geq 0$, the $k$-neighbourhood of $v$, also known as the ball of center $v$ and radius $k$, is defined as $N_G^k[v] = \{ u \in V \mid dist_G(u,v) \leq k \}$.
For instance, $N_G^1[v]$ is exactly the closed neighbourhood of $v$.
The diameter of $G$ is equal to $diam(G) = \max_{u,v \in V} dist_G(u,v)$.

\begin{center}
	\fbox{
		\begin{minipage}{.95\linewidth}
			\begin{problem}[\textsc{Diameter}]\
				\label{prob:diam} 
					\begin{description}
					\item[Input:] A graph $G=(V,E)$.
					\item[Output:] The diameter of $G$.
				\end{description}
			\end{problem}     
		\end{minipage}
	}
\end{center}

\begin{theorem}[\cite{RoV13}]
Under the Strong Exponential-Time Hypothesis, we cannot decide whether a graph has diameter at most two in time ${\cal O}(mn^{1-\varepsilon})$, for any $\varepsilon > 0$.
\end{theorem}

\subsection{Hypergraphs}\label{sec:hypgraph}

More generally, a hypergraph is a pair ${\cal H} = (X,R)$ with $X$ being the set of vertices and $R \subseteq 2^X$ being the set of hyperedges.
See also~\cite{Ber73} for any undefined hypergraph terminology.
Let $n := |X|$, $m := \sum_{q \in R} |q|$ and $r := |R|$ be the order, the size and the number of hyperedges of ${\cal H}$, respectively.
For every vertex $x \in X$, let $R_x := \{ q \in R \mid x \in q \}$.
The {\em dual} of ${\cal H}$ is the hypergraph ${\cal H}^* := (R,X^*)$, where $X^* := \{ R_x \mid x \in X \}$.
In particular, ${\cal H}$ and ${\cal H}^{**}$ are isomorphic.

Several hypergraphs can be related to a graph $G$:
\begin{itemize}
\item The closed neighbourhood hypergraph, denoted by ${\cal N}_1(G)$, has vertex-set $X = V$ and hyperedge-set $R = \{ N_G[v] \mid v \in V\}$;
\item More generally, for every fixed $\ell \geq 0$, the $\ell$-neighbourhood hypergraph of $G$ is defined as ${\cal N}_{\ell}(G) = (V, \{ N_G^{\ell}[v] \mid v \in V \})$.
We stress that ${\cal N}_{\ell}(G)$ and its dual ${\cal N}_{\ell}^*(G)$ are isomorphic~\cite{BoT15}. 
\item Finally, the ball hypergraph of $G$, simply denoted by ${\cal B}(G)$, has for hyperedges the balls of all possible centers and radii in $G$.
Equivalently, ${\cal B}(G) = \bigcup_{\ell \geq 0} {\cal N}_{\ell}(G)$.
\end{itemize}

\subsection{VC-dimension}\label{sec:vc-dim}
Let ${\cal H} = (X,R)$ be a fixed hypergraph.
A subset $Y \subseteq X$ is {\em shattered} by ${\cal H}$ if, for every $Y' \subseteq Y$, there exists a hyperedge $q \in R$ such that $Y \cap q = Y'$.
Then, the Vapnik-Chervonenkis dimension of ${\cal H}$ (abbreviated in what follows to {\em VC-dimension}) is the largest cardinality of a shattered subset.
Similarly, the {\em dual VC-dimension} of ${\cal H}$ is the VC-dimension of its dual ${\cal H}^*$.
We will often use the following (easy) properties in our analysis: 

\begin{lemma}[Sauer-Shelah-Perles,~\cite{Sau72,She72}]\label{lem:sauer-lemma}
Every $n$-vertex hypergraph of VC-dimension at most $d$ has ${\cal O}(n^d)$ hyperedges.
\end{lemma}

\begin{lemma}[\cite{ChW89}]\label{lem:dual-vc}
Every hypergraph of VC-dimension $d$ has dual VC-dimension at most $2^{d+1}$.
\end{lemma}

\begin{lemma}[\cite{Kle04}]\label{lem:vc-sub}
For every hypergraph ${\cal H} = (X,R)$ and $Y \subseteq X$, let $R[Y] = \{ q \cap Y \mid q \in R \}$.
Then, the VC-dimension of ${\cal H}[Y] := (Y,R[Y])$ is at most the VC-dimension of ${\cal H}$.
\end{lemma}

\paragraph{VC-dimension for graphs.}
The {\em VC-dimension} of a graph $G$ is defined as the VC-dimension of its closed neighbourhood hypergraph ${\cal N}_1(G)$.
For instance, $K_h$-minor free graphs (and so, $H$-minor free graphs for any $H$ of order at most $h$) have VC-dimension at most $h-1$~\cite{ABC95}.
Every $k$-interval graph has VC-dimension in ${\cal O}(k\log{k})$~\cite{DHV19+}.
Other classes of constant VC-dimension -- at most three -- are unit disk graphs, chordal bipartite graphs, $C_4$-free bipartite graphs, graphs of girth at least five and undirected path graphs~\cite{BLLP+15}. 

\smallskip
\noindent
The {\em distance VC-dimension} of a graph $G$ is defined as the VC-dimension of its ball hypergraph ${\cal B}(G)$. Chepoi, Estellon and Vax{\`e}s proved in~\cite{CEV07} that planar graphs have distance VC-dimension at most $4$, and remarked that more generally every $K_h$-minor free graph has distance VC-dimension at most $h-1$.
Bousquet and Thomass\'e proved in~\cite{BoT15} that graphs of bounded distance VC-dimension also generalize graphs of bounded rankwidth.
Indeed, every graph of rankwidth $k$ has distance VC-dimension at most $3 \cdot 2^{k+1} + 1$.
For purpose of illustration, we next adapt a proof from~\cite{BLLP+15} in order to show that interval graphs have distance VC-dimension at most two:

\begin{lemma}\label{lem:interval}
Every interval graph has distance VC-dimension at most $2$.
\end{lemma}

\begin{proof}
Let $G=(V,E)$ be an interval graph.
We fix an interval model for $G$.
For every $v \in V$, let $I(v) = [a_v,b_v]$ be the corresponding interval in the representation.
Suppose now by contradiction that there is a set $S = \{v_1,v_2,v_3\}$ that is shattered by ${\cal B}(G)$.
W.l.o.g., $a_{v_1} < a_{v_2} < a_{v_3}$.
Since $S$ is shattered, there exist some $u \in V$ and $k \geq 0$ such that $N_{G}^k[u] \cap S = \{v_1,v_3\}$.
But then, let $I_{k-1}(u) := \bigcup_{w \in N_G^{k-1}[u]}I(w)$ be the contiguous segment of all the vertices at a distance $\leq k-1$ from $u$.
Note that $I_{k-1}(u) \cap I(v_2) = \emptyset$ because we assume that $v_2 \notin N_{G}^k[u]$.
In this situation, either $I_{k-1}(u) \subseteq ]-\infty,a_{v_2}[$ or $I_{k-1}(u) \subseteq ]b_{v_2},\infty[$ where $]x,y[\ =[x,y]\setminus\{x,y\}$ denotes the open interval between $x$ and $y$. 
In fact we must have $I_{k-1}(u) \subseteq ]b_{v_2},\infty[$ because otherwise, $I_{k-1}(u) \cap I(v_3) = \emptyset$ and so, $v_3 \notin N_G^k[u]$, a contradiction.
Since $I_{k-1}(u) \cap I(v_1) \neq \emptyset$, it implies that $b_{v_1} > b_{v_2}$, and so, $I(v_2) \subseteq I(v_1)$.
As a result we have $N_{G}[v_2] \subseteq N_G[v_1]$.
But then, for any $w \in V$ and $\ell \geq 1$, we have $v_2 \in N_G^{\ell}[w] \Longrightarrow v_1 \in N_G^{\ell}[w]$.
The latter contradicts our hypothesis that $S$ is shattered.
\end{proof}

\subsection{Stabbing number and applications to {\sc Diameter}}\label{sec:stabbing-num}
A spanning tree of ${\cal H} = (X,R)$ is a tree $T$ whose node-set is exactly $X$.
The {\em stabbing number} of such spanning tree $T$ is the least $k$ such that, for every hyperedge $q \in R$, there exist at most $k$ edges $uv \in E(T)$ such that $|q \cap \{u,v\}| = 1$ (we also say that $uv$ is {\em stabbed} by $q$). Given a set $q\subseteq X$, we let $E_T(q) = \{ uv \in E(T) \mid u \in q, v \notin q \}$ of all edges stabbed by $q$. 
Finally, the stabbing number of ${\cal H}$ is the minimum stabbing number over its spanning {\em paths}. 
Indeed, as noted in~\cite{ChW89}, every spanning tree $T$ can be transformed into a spanning path of stabbing number at most twice bigger than for $T$. Therefore, there is essentially no loss of generality in restricting ourselves to spanning paths.

\begin{lemma}[\cite{ChW89}]\label{lem:stabbing-num}
Every $n$-vertex hypergraph of dual VC-dimension $d$ has stabbing number $\tilde{\cal O}(n^{1-\frac 1 d})$.
\end{lemma}

Overall it follows from Lemmata~\ref{lem:dual-vc} and~\ref{lem:stabbing-num} that any $n$-vertex hypergraph of VC-dimension at most $d$ has strongly sublinear stabbing number in $\tilde{\cal O}(n^{1-\frac 1 {2^{d+1}}})$.  
We stress that the proof of Lemma~\ref{lem:stabbing-num} is constructive but that it cannot be transformed into a truly subquadratic-time algorithm.
Efficient computations of spanning paths with sublinear stabbing number -- or related data structures -- were proposed for many special cases from computational geometry~\cite{Cha12,Mat91,Wel92}.

\begin{center}
	\fbox{
		\begin{minipage}{.95\linewidth}
			\begin{problem}[\textsc{$f$-Approx Stabbing Number}]\
				\label{prob:stabbing-num} 
					\begin{description}
					\item[Input:] A hypergraph ${\cal H}=(X,R)$ of VC-dimension at most $d$.
					\item[Output:] A spanning path $P$ of stabbing number at most $\tilde{\cal O}(n^{1-\frac 1 {f(d)}})$ and, for every $q \in R$, the set $E_P(q) = \{ uv \in E(P) \mid u \in q, v \notin q \}$ of all edges stabbed by $q$. 
				\end{description}
			\end{problem}     
		\end{minipage}
	}
\end{center}

{\em For simplicity of exposition, we will assume throughout the remainder of this paper that the VC-dimension of all the hypergraphs considered is part of the input.}
However in practice, we can easily weaken this assumption as follows.
Given some ``guess'' $d$ on the VC-dimension of the input, we can modify our proposed solutions so that they either output a spanning path whose stabbing number is at most $\tilde{\cal O}(n^{1-\frac 1 {f(d)}})$, for some function $f$, or conclude that the VC-dimension of the input is larger than $d$.
By dichotomic search, we so can compute some minimum $d^*$ such that, for any $d \geq d^*$, our algorithms always output a spanning path of stabbing number $\tilde{\cal O}(n^{1-\frac 1 {f(d)}})$.
We stress that $d^*$ is at most the VC-dimension of $G$, but that it can be much smaller in practice.

\paragraph{Reduction from diameter computation.}
We now recall the following simple but beautiful approach that we use in order to solve {\sc Diameter} on graphs of constant VC-dimension.

\begin{lemma}\label{lem:reduction}
Let $G$ be a graph and $k \geq 2$.
If the hypergraph ${\cal N}_{k-1}(G)$ has VC-dimension at most $d$, and we can solve \textsc{$f$-Approx Stabbing Number} for ${\cal N}_{k-1}(G)$ in time $T(n,m)$, then we can decide whether $G$ has diameter at most $k$ in time $\tilde{\cal O}(T(n,m) + mn^{1-\frac{1}{f(d)}})$.
\end{lemma}

\begin{proof}
Let us first compute a spanning path $P$ of stabbing number at most $\tilde{\cal O}(n^{1-\frac 1 {f(d)}})$ for ${\cal N}_{k-1}(G)$.
By the hypothesis, it takes ${\cal O}(T(n,m))$ time.
For every $v \in V$ we can compute from $E_P(N_G^{k-1}[v])$ a set $I_{k-1}(v)$ of $t_v$ intervals, where $|E_P(N_G^{k-1}[v])| -1 \leq t_v \leq |E_P(N_G^{k-1}[v])| + 1$, such that $\bigcup I_{k-1}(v) = N_G^{k-1}[v]$.
This preprocessing phase takes time ${\cal O}(|E_P(N_G^{k-1}[v])|) = \tilde{\cal O}(n^{1-\frac 1 {f(d)}})$, and so, $\tilde{\cal O}(n^{2-\frac 1 {f(d)}})$ total time.
Then in order to decide whether $diam(G) \leq k$, we are left to decide whether for every $u \in V$ we have $\bigcup_{v \in N_G[u]} I_{k-1}(v) = V$.
For that, it suffices to collect the $\tilde{\cal O}(deg_G(u) \cdot n^{1-\frac 1 {f(d)}})$ ends of the intervals in $\bigcup_{v \in N_G[u]} I_{k-1}(v)$, and then to order them lexicographically.
As a result, this last verification phase can be done in total time $\tilde{\cal O}(mn^{1-\frac 1 {f(d)}})$.
\end{proof}

\begin{remark}\label{rk:radius}
The {\em radius} of a graph $G$ is equal to $rad(G) = \min_{u \in V}\max_{v \in V} dist_G(u,v)$.
Under the Hitting Set conjecture, we cannot compute the radius of a graph in truly subquadratic-time~\cite{AVW16}.
We here observe that we can easily modify the framework of Lemma~\ref{lem:reduction} in order to decide whether a graph has radius at most $k$.
Indeed, for that it suffices to check whether there exists at least one vertex $u$ such that $\bigcup_{v \in N_G[u]} I_{k-1}(v) = V$.
\end{remark}

Our main task in the remainder of this article will be to solve \textsc{$f$-Approx Stabbing Number} efficiently on $\ell$-neighbourhood hypergraphs, for some 
increasing function $f$.
Then, we can apply Lemma~\ref{lem:reduction} in order to efficiently solve {\sc Diameter}.

\section{Computation of Spanning paths with low Stabbing Number}\label{sec:spanning-path}

We prove in this section our first main result in the paper, whose statement is reminded below:

\diamTwo*

We will need the following result in our proofs:


\begin{lemma}[\cite{ChW89}]\label{lem:approx-stabbing-number}
There is a deterministic polynomial-time algorithm that outputs, for every $n$-vertex hypergraph ${\cal H}$ of VC-dimension at most $d$, a spanning path of stabbing number ${\cal O}(n^{1-1/2^{d+1}}\log n)$.
\end{lemma}

This above lemma is a consequence of Theorem~4.3 in~\cite{ChW89} and the discussion about the complexity of the algorithm resulting from their proof.
We note that their result applies to infinite range spaces too, with the initial step in their proof reducing to the finite case.
In order to derive Lemma~\ref{lem:approx-stabbing-number} from~\cite{ChW89}, we use the bound on the dual VC-dimension resulting from Lemma~\ref{lem:dual-vc} and the fact that no initial step is required as we start from a {\em finite} range space.
Better randomized algorithms can be obtained through the approximation results in~\cite{BGRS04,Har09}. They are expressed for spanning trees but easily convert to paths as previously noted.
The algorithms in~\cite{BGRS04,Har09} use LP relaxation and randomized rounding. It is not immediately clear if they can be derandomized using classical techniques. Indeed, the algorithm from~\cite{Har09} works by phases. During a phase, it needs to solve an ILP relaxation and then to apply some randomized rounding technique. In the worst case, this main phase is repeated ${\cal O}(\log{n})$ times.
We observe that even by using the best known upper-bounds on the time complexity of linear programming~\cite{CLS19}, this overall process takes super-quadratic time.
In what follows, we use the Sauer-Shelah-Perles Lemma (Lemma~\ref{lem:sauer-lemma}) in order to obtain better trade-offs between the running-time and the quality of our solution.

\stabbingNb*

\begin{proof}
Let $\eta \in (0;1)$ to be fixed later in the proof.
We arbitrarily partition the vertex-set $X$ into subsets $X_1,X_2,\ldots,X_p$ such that $p = {\cal O}(n^{1-\eta})$ and, for every $1 \leq i \leq p$, $|X_i| = {\cal O}(n^{\eta})$.
Our aim is to apply Lemma~\ref{lem:approx-stabbing-number} to the subhypergraphs ${\cal H}[X_1], {\cal H}[X_2], \ldots, {\cal H}[X_p]$.
We stress that all these subhypergraphs can be constructed in total ${\cal O}(m)$-time, as follows: we scan all the hyperedges $q$ once in order to compute $(q \cap X_i)_{1 \leq i \leq p}$; then, for every $i$, we use a linear-time sorting algorithm in order to suppress duplicated values in $\{ q \cap X_i \mid q \in R \}$.

\begin{myclaim}\label{claim:stabbing-partition}
Given ${\cal H}[X_1], {\cal H}[X_2], \ldots, {\cal H}[X_p]$, we can compute a spanning path for ${\cal H}$ of stabbing number $\tilde{\cal O}(n^{1-\frac{\eta}{2^{d+1}}})$.
Moreover, it takes ${\cal O}(n^{1 + \eta[c(d+1)-1]})$ time for some universal constant $c > 2$.
\end{myclaim}

\begin{proofclaim}
By Lemma~\ref{lem:vc-sub}, every ${\cal H}[X_i]$ has VC-dimension at most $d$.
This implies that ${\cal H}[X_i]$ has ${\cal O}(n^{\eta d})$ hyperedges (Lemma~\ref{lem:sauer-lemma}), and so it has size ${\cal O}(n^{\eta(d+1)})$.
Furthermore 
by Lemma~\ref{lem:approx-stabbing-number} we can compute deterministically a spanning path of stabbing number $\tilde{\cal O}\left(n^{\eta \left(1 - \frac 1 {2^{d+1}} \right)}\right)$, in time ${\cal O}(n^{c \eta (d+1)})$ for some universal constant $c$.

Let $P_1, P_2, \ldots, P_p$ be the spanning paths that we obtain.
We obtain a spanning path $P$ for ${\cal H}$ by concatenating all the $P_i$'s.
For every $1 \leq i \leq p$, we recall that the stabbing number of $P_i$ is in $\tilde{\cal O}\left(n^{\eta \left(1 - \frac 1 {2^{d+1}} \right)}\right)$.
Therefore by construction, the stabbing number of $P$ is in $\tilde{\cal O}\left( p \cdot n^{\eta \left(1 - \frac 1 {2^{d+1}} \right)} + p -1 \right) = \tilde{\cal O}\left(n^{1 - \frac{\eta}{2^{d+1}}}\right)$.
\end{proofclaim}

Let $P$ be the spanning path obtained with Claim~\ref{claim:stabbing-partition}.
Finally, for every $q \in R$ we compute the set $E_P(q)$ of all edges of $P$ stabbed by $q$, in total ${\cal O}(m)$-time, simply by scanning once all the hyperedges.
The total running-time is in ${\cal O}(m + p \cdot n^{c \eta (d+1)}) = {\cal O}(m + n^{1 + \eta[c(d+1)-1]})$.
Overall, we achieve a good trade-off between running-time and approximation factor if we have $2 - \frac{\eta}{2^{d+1}} = 1 + \eta[c(d+1)-1]$.
Therefore we set $\eta = \frac 1 {c(d+1) + \frac 1 {2^{d+1}} - 1}$, and then $\varepsilon_d = \frac{\eta}{2^{d+1}} = \frac 1 {2^{d+1}[ c(d+1) - 1 ] + 1}$.
\end{proof}

We observe that our analysis could be easily improved in some particular cases, {\it e.g.}, for all hypergraphs that are isomorphic to their dual.

We are now ready to prove the main result in this section:

\begin{proofof}{Theorem~\ref{thm:diamTwo}}
We apply Theorem~\ref{thm:stabbingNb} to the closed neighbourhood hypergraph of $G$.
Then, the result follows from Lemma~\ref{lem:reduction} applied to the function $f : d \to 1/\varepsilon_d$.
\end{proofof}

\section{Bounded {\sc Diameter} with $\varepsilon$-nets}\label{sec:diam-comput}

For graphs of bounded {\em distance} VC-dimension we now generalize Theorem~\ref{thm:diamTwo} from the previous section to larger values for the diameter.

\boundedDiam*

Our proof crucially relies on the concept of {\em $\varepsilon$-net}.
We recall that for a hypergraph ${\cal H} =(X,R)$, a subset $Y \subseteq X$ is called an $\varepsilon$-net if, for every $q \in R$, we have $|q| \geq \varepsilon n \Longrightarrow Y \cap q \neq \emptyset$.

\begin{lemma}[\cite{HaW87,VaC15}]\label{lem:eps-net}
For every hypergraph of VC-dimension at most $d$, any random subset of size ${\Omega}\left( \frac d {\varepsilon} \log{\left(\frac{1}{\varepsilon\delta}\right)} \right)$ is an $\varepsilon$-net with probability $1 - \delta$.
\end{lemma}

We will also need the following result:

\begin{lemma}[\cite{ChW89}]\label{lem:diff}
For every hypergraph ${\cal H} = (X,R)$, let $\hat{R} := \{ q \Delta q' \mid q,q' \in R\}$ be the set of symmetric differences between hyperedges.
If ${\cal H}$ has VC-dimension at most $d$ then, $\hat{\cal H} := (X,\hat{R})$ has bounded VC-dimension.
\end{lemma}

We observe that no explicit upper bound on the VC-dimension of $\hat{\cal H}$ was stated in~\cite{ChW89}. 
Nevertheless it can be easily deduced from their proof that it is in ${\cal O}(d\log{d})$ (see also~\cite{EiA07}).

\smallskip
The following partition lemma which derives from the two above lemmata is the cornerstone of our algorithm.

\begin{lemma}\label{lem:partition-lemma}
Let $G=(V,E)$ be a graph of distance VC-dimension at most $d$, and let $S$ be any random subset of size $\tilde{\Theta}(d/\varepsilon)$.
Then w.h.p., for every $\ell \geq 0$ and for every $u,v \in V$ such that $N_G^{\ell}[u] \cap S = N_G^{\ell}[v] \cap S$, we have $|\ N_G^{\ell}[u] \Delta N_G^{\ell}[v] \ | = \tilde{\cal O}(\varepsilon n)$.
\end{lemma}

\begin{proof}
Let $\hat{R} = \{ \ N_G^{\ell_1}[x] \Delta N_G^{\ell_2}[y] \ \mid x,y \in V \ \mbox{and} \ \ell_1,\ell_2 \geq 0 \}$ be the set of the symmetric differences between the balls of $G$.
Since $G$ has distance VC-dimension at most $d$ then, by Lemma~\ref{lem:diff}, the hypergraph $\hat{\cal H} = (V,\hat{R})$ has VC-dimension in ${\cal O}(d\log{d})$.
Then by Lemma~\ref{lem:eps-net}, w.h.p. $S$ is an $\varepsilon$-net for $\hat{\cal H}$.
Therefore, for every $\ell \geq 0$ and for every $u,v \in V$, $|\ N_G^{\ell}[u] \Delta N_G^{\ell}[v] \ | > \varepsilon n \Longrightarrow  (N_G^{\ell}[u] \Delta N_G^{\ell}[v]) \cap S \neq \emptyset$. We stress that $(N_G^{\ell}[u] \Delta N_G^{\ell}[v]) \cap S \neq \emptyset \Longrightarrow N_G^{\ell}[u] \cap S \neq N_G^{\ell}[v] \cap S$. 
\end{proof}

This above partition lemma will be useful in order to group the vertices in a small number of groups, with every two vertices in a group having almost the same ball of radius $\ell$.
Here there is a trade-off between the number of groups (that we upper-bound by using the Sauer-Shelah-Perles Lemma) and, for every two vertices in the same group, the maximum number of vertices in which their respective balls of radius $\ell$ can differ. 

\smallskip
More precisely, our approach in the next two sections can be summarized as follows:
\begin{enumerate}
\item We compute a spanning path $P_k'$ for ${\cal N}_k(G)$ of low {\em average} stabbing number, with the latter being equal to $\frac 1 n \cdot \sum_{v \in V} |E_{P_k'}(N_G^k[v])|$;
\item Then, we compute an $\varepsilon$-net, for some well-chosen $\varepsilon$, and in doing so we partition the vertex-set into $p(\varepsilon)$ disjoint groups $V_1, V_2,\ldots,V_{p(\varepsilon)}$. For every $j$ we select a unique $v_j \in V_j$. We restrict ourselves to ${\cal H}_k := (V, \{ N_G^k[v_j] \mid 1 \leq j \leq p(\varepsilon) \})$. We compute a spanning path $P_k$ of low stabbing number for this subhypergraph.
\item We observe that if $P_k$ is a spanning path of stabbing number $t$ for ${\cal H}_k$, then it is also a spanning path of stabbing number $t + {\cal O}(\varepsilon n)$ for ${\cal N}_k(G)$. Finally, for every $1 \leq j \leq p(\varepsilon)$, we consider the unselected vertices $u \in V_j \setminus \{v_j\}$ sequentially. We compute the set of all the edges in $E(P_k)$ that are stabbed by $N_G^k[u]$. For that, it suffices to compute the ${\cal O}(\varepsilon n)$ vertices of $N_G^k[u] \Delta N_G^k[v_j]$. We do so efficiently by using the auxiliary spanning path $P_k'$.
\end{enumerate}

We next give a first application of our approach (we will give another such application in the proof of Theorem~\ref{thm:subLin}).

\begin{proofof}{Theorem~\ref{thm:boundedDiam}}
Let $\varepsilon_d$ be the constant of Theorem~\ref{thm:stabbingNb}.
We shall prove the following claim by finite induction: 

\begin{myclaim}\label{claim:+1}
For every $1 \leq i \leq k-1$, we can compute a spanning path of stabbing number $\tilde{\cal O}(n^{1 - \varepsilon_d})$ for ${\cal N}_i(G)$. Moreover, it can be done in time $\tilde{\cal O}(i \cdot mn^{1-\varepsilon_d})$.
\end{myclaim}

The proof of Theorem~\ref{thm:boundedDiam} will follow from this claim and Lemma~\ref{lem:reduction} by taking $i = k-1$. 

\begin{proofclaim}
By Theorem~\ref{thm:stabbingNb}, the claim is true for the base case $i=1$. 
Assume by our induction hypothesis that the claim holds for $i-1$.
We divide the remainder of the proof into two subclaims.

\begin{subclaim}\label{subclaim:subopt}
Let $P_{i-1}$ be a spanning path of stabbing number $t$ for ${\cal N}_{i-1}(G)$.
We can transform $P_{i-1}$ into a spanning path $P_i'$ for ${\cal N}_i(G)$, such that $\sum_{v \in V} |E_{P_i'}(N_G^i[v])| = {\cal O}(tm)$.
Moreover, the transformation takes time ${\cal O}(tm)$.
\end{subclaim}

\begin{proofsubclaim}
Let $u \in V$.
Then in time ${\cal O}(deg_G(u) \cdot t)$, we can collect the edge-sets $E_{P_{i-1}}(N_G^{i-1}[w])$ of all the edges of $P_{i-1}$ that are stabbed by $w$, for $w \in N_G[u]$. We compute from these edge-sets a (suboptimal) representation of $N_G^i[u]$ into ${\cal O}(deg_G(u) \cdot t)$ intervals of $P_{i-1}$.
\end{proofsubclaim}

\begin{subclaim}\label{subclaim:apply-eps-net}
Let $P_i'$ be a spanning path for ${\cal N}_i(G)$, such that $\sum_{v \in V} |E_{P_i'}(N_G^i[v])| = {\cal O}(tm)$.
Then, in time $\tilde{\cal O}((n^{1 - \varepsilon_d} + t)\cdot m)$, we can compute a spanning path $P_i$ of stabbing number $\tilde{\cal O}(n^{1-\varepsilon_d})$.
\end{subclaim}

\begin{proofsubclaim}
Let $\varepsilon := \Theta(n^{-\varepsilon_d})$.
We perform a breadth-first search from every vertex in some random subset $S$ of cardinality $\tilde{\cal O}(d/\varepsilon) = \tilde{\cal O}(n^{\varepsilon_d})$.
In doing so we define an equivalence relation $\sim$ on $V$ such that $u \sim v \Longleftrightarrow^{def} N_G^i[u] \cap S = N_G^i[v] \cap S$.
We so partition $V$ into some groups $V_1, V_2, \ldots, V_p$.
Since by the hypothesis $G$ has distance VC-dimension at most $d$ then, by Lemma~\ref{lem:sauer-lemma} we have $p = {\cal O}(|S|^d) = \tilde{\cal O}(n^{\varepsilon_dd})$. 
Furthermore by Lemma~\ref{lem:partition-lemma}, we have w.h.p. $u \sim v \Longrightarrow |\ N_G^{i}[u] \Delta N_G^{i}[v] \ | = \tilde{\cal O}(\varepsilon n) = \tilde{\cal O}(n^{1-\varepsilon_d})$.
The algorithm now proceeds as follows:
\begin{enumerate}
\item For every $1 \leq j \leq p$, we select a unique $v_j \in V_j$, and then we start a breadth-first search from this vertex.  Since $p = \tilde{\cal O}(n^{\varepsilon_dd})$ and we have $\varepsilon_d \ll 1/d$, this phase can be implemented in time $\tilde{\cal O}(mn^{\varepsilon_d d}) = \tilde{o}(mn^{1-\varepsilon_d})$, that is truly subquadratic.
\item Let $R_i := \{ N_G^i[v_j] \mid 1 \leq j \leq p \}$, and let ${\cal H}_i := (V,R_i)$.
Note that since ${\cal H}_i \subseteq {\cal B}(G)$, the VC-dimension of ${\cal H}_i$ is at most $d$.
Furthermore, the order and size of ${\cal H}_i$ are, respectively, $n$ and $m_i := {\cal O}(pn) = \tilde{\cal O}(n^{1+\varepsilon_d d})$.
By Theorem~\ref{thm:stabbingNb}, we can compute a spanning path $P_i$ for ${\cal H}_i$ of stabbing number $\tilde{\cal O}(n^{1-\varepsilon_d})$ in time $\tilde{\cal O}(m_i + n^{2-\varepsilon_d}) = \tilde{\cal O}(n^{1+\varepsilon_d d} + n^{2-\varepsilon_d}) = \tilde{\cal O}(n^{1-\varepsilon_d}m)$.
\item We observe that $P_i$ is a spanning path of ${\cal N}_i(G)$ of stabbing number:
$$\tilde{\cal O}(n^{1-\varepsilon_d}) + \max\limits_{1 \leq j \leq p}\max_{u \in V_j \setminus \{v_j\}} | \ N_G^i[u] \Delta N_G^i[v_j] \ | =  \tilde{\cal O}(n^{1-\varepsilon_d}).$$
We are now left with computing, for every $1 \leq j \leq p $ and $u \in V_j \setminus \{v_j\}$, the set $E_{P_i}(N_G^i[u])$ of all the edges stabbed by the ball of radius $i$ centered at $u$.
For that, since we are already given $E_{P_i}(N_G^i[v_j])$, it suffices to compute $N_G^i[u] \Delta N_G^i[v_j]$.
We proceed in three steps:
\begin{itemize}
\item By our hypothesis, we computed a spanning path $P_i'$ for ${\cal N}_{i}(G)$, such that $\sum_{u \in V} |E(N_G^i[u])| = {\cal O}(tm)$.
Then, we can compute from $P_i'$ a (suboptimal) representation $I_i(u)$ of $N_G^i[u]$ into ${\cal O}(|E_{P_i'}(N_G^i[u])|)$ intervals. 
In doing so, we also compute within the same amount of time a representation $\overline{I_i(u)}$ of $V \setminus N_G^i[u]$ into ${\cal O}(|E_{P_i'}(N_G^i[u])|)$ intervals of $P_{i}'$. 
Overall this step takes total time $\tilde{\cal O}(tm)$.
\item Let $\sigma_i : V \to V(P_{i}')$ be the permutation that maps every vertex to its position in the spanning path $P_{i}'$. For every $1 \leq j \leq p$, we construct two balanced binary search trees whose items are, respectively, $\{ \sigma_{i}(x) \mid x \in N_G^i[v_j]  \}$ and $\{ \sigma_{i}(y) \mid y \notin N_G^i[v_j] \}$.
Overall, this takes total time $\tilde{\cal O}(np) = \tilde{\cal O}(n^{1+\varepsilon_d d}) = \tilde{o}(mn^{1-\varepsilon_d})$. 
\item Finally, let us again consider some $u \in V_j \setminus \{v_j\}$ for some $j$. For every interval from $I_i(u)$, we want to enumerate the vertices of $V \setminus N_G^i[v_j]$ that lie on this interval. Since we stored all of $V \setminus N_G^i[v_j]$ into a balanced binary search tree, this can be done in time ${\cal O}(\log{n})$ plus ${\cal O}(1)$ extra time per solution.
In the same way, for every interval from $\overline{I_i(u)}$, we enumerate the vertices of $N_G^i[v_j]$ that lie on this interval. For a fixed $u$, the total time for this step is in $\tilde{\cal O}( \ |I_i(u)| + |\overline{I_i(u)}| + |N_G^i[u]\Delta N_G^i[v_j]| \ ) = \tilde{\cal O}( |E_{P_i'}(N_G^i[u])| + n^{1 - \varepsilon_d})$. 
Therefore, this last step takes total time $\tilde{\cal O}(tm + n^{2-\varepsilon_d})$.
\end{itemize}
\end{enumerate}
\end{proofsubclaim}

Now, by the induction hypothesis we get a spanning path of stabbing number $\tilde{\cal O}(n^{1-\varepsilon_d})$ for ${\cal N}_{i-1}(G)$.
By Subclaim~\ref{subclaim:subopt} we transform such spanning path into a spanning path $P_i'$ for ${\cal N}_{i}(G)$, where $\sum_{u \in V} |E_{P_i'}(N_G^i[u])| = \tilde{\cal O}(mn^{1-\varepsilon_d})$.
Finally, by Subclaim~\ref{subclaim:apply-eps-net} we can use $P_i'$ in order to compute, in time $\tilde{\cal O}(mn^{1-\varepsilon_d})$, a spanning path $P_i$ of stabbing number $\tilde{\cal O}(n^{1-\varepsilon_d})$.
The above algorithm achieves proving that our claim holds for $i$.
\end{proofclaim}

Summarizing, by Claim~\ref{claim:+1} we can compute a spanning path of stabbing number $\tilde{\cal O}(n^{1-\varepsilon_d})$ for the hypergraph ${\cal N}_{k-1}(G)$, in time $\tilde{\cal O}(k \cdot mn^{1-\varepsilon_d})$.
By Lemma~\ref{lem:reduction} it implies that we can also decide whether $G$ has diameter at most $k$, and if so compute $diam(G)$ exactly, in time $\tilde{\cal O}(k \cdot mn^{1-\varepsilon_d})$.
\end{proofof}

\subsection{Application to nowhere dense graph classes}\label{sec:no-dense}
A closer look at the proof of Theorem~\ref{thm:boundedDiam} shows that it also holds if, instead of having bounded distance VC-dimension, there rather exists some constant $d$ such that, for every $1 \leq i \leq k-1$, the VC-dimension of the $i$-neighbourhood hypergraph is at most $d$ (the latter value is sometimes called the distance-$i$ VC-dimension of the graph~\cite{NeO16}).
It has algorithmic implications for some special cases of sparse graphs.
Namely, $H$ is an {\em $r$-shallow minor} of a graph $G$ if it can be obtained from some subgraph of $G$ by the contraction of pairwise disjoint subgraphs of radius at most $r$~\cite{PRS94}; a graph family ${\cal G}$ is termed {\em nowhere dense} if, for any $r$, there exists a graph $H_r$ which is not an $r$-shallow minor for any graph in ${\cal G}$~\cite{NeO12}. 
Of interest here is that, for any graph class ${\cal G}$ nowhere dense, and for any $i$, the distance-$i$ VC-dimension of any graph in ${\cal G}$ is upper-bounded by some constant $d_i$~\cite{NeO16}.
By choosing $d := \max_{1 \leq i \leq k-1} d_i$, we thus obtain the following weaker version of Theorem~\ref{thm:boundedDiam} for nowhere dense graphs:

\nowhereDense*

We left open whether there exists a truly subquadratic-time {\em FPT} algorithm for diameter computation on nowhere dense graph classes (i.e., with no dependency on $k$ in the exponent).

\section{Diameter computation in truly Subquadratic time}\label{sec:refinements}

We finally improve the results of Theorem~\ref{thm:boundedDiam} for a more restricted family of graphs of bounded distance VC-dimension.
Before that, we need to introduce a bit more of graph terminology. 
A class of graphs is called {\em monotone} if it is closed by taking subgraphs.
For a connected $n$-vertex graph $G$, a {\em separator} is a subset $S$ such that $G \setminus S$ is disconnected. 
It is called {\em balanced} if every connected component of $G \setminus S$ has order at most $2n/3$.
Finally, a class of graphs has {\em strongly sublinear} balanced separator if every connected $n$-vertex graph in the class has a balanced separator of cardinality at most $C \cdot n^{\alpha}$ for some constants $C$ and $\alpha < 1$.

\subLin*

We postpone the technical proof of this result to Sec.~\ref{sec:proof-sublin}.
Let us emphasize that Theorem~\ref{thm:subLin} cannot be applied to {\em all} graph classes of bounded distance VC-dimension.
For instance, we proved in Lemma~\ref{lem:interval} that the intervals graphs have distance VC-dimension at most two.
However, there exist intervals graphs with no balanced separators of sublinear size. 
We give some interesting cases where Theorem~\ref{thm:subLin} {\em does} apply in Sec.~\ref{sec:applications-sublin}.

\medskip
Finally, we say that a class of graphs ${\cal G}$ has {\em polynomial expansion} if there exists a polynomial $p$ such that, for every $r$-shallow minor of a graph in ${\cal G}$ (cf. Section~\ref{sec:no-dense}), the average degree is at most $p(r)$. 
We want to stress that there is an equivalence between the monotone classes of graphs ${\cal G}$ with strongly sublinear balanced separators and those of polynomial expansion~\cite{DvN16}.
In particular, the graphs in ${\cal G}$ have bounded degeneracy, and so, they are sparse (i.e., with $m = {\cal O}(n)$ edges).
We will often use this property in what follows.  

\subsection{Application to $H$-minor free graphs}\label{sec:applications-sublin}

Let us now review some interesting classes where Theorem~\ref{thm:subLin} does apply.
Since planar graphs have distance VC-dimension at most four~\cite{BoT15} then, it follows from the planar separator theorem of Lipton and Tarjan~\cite{LiT79} that it is the case for planar graphs.
Therefore, Theorem~\ref{thm:subLin} gives us a new subquadratic-time algorithm for diameter computation on {\em unweighted} planar graphs, but with a slower running-time than for the algorithms presented in~\cite{Cab18,GKHM+18}. 
More generally, the following separator theorem is from Alon et al.:

\begin{lemma}[\cite{AST90}]\label{lem:h-minor}
Every $K_h$-minor free graph has a balanced separator of cardinality ${\cal O}(h^{3/2}\sqrt{n})$.
Moreover, such a separator can be found in ${\cal O}(n^{3/2})$ time.
\end{lemma}

See also~\cite{KaR10,Wul11} for various trade-offs between the size of the separator and the time that is needed in order to find it. 
We recall that $K_h$-minor free graphs have distance VC-dimension at most $h-1$~\cite{BoT15}.
By combining this result with Lemma~\ref{lem:h-minor}, we so prove the following meta-theorem:

\begin{corollary}\label{cor:h-minor}
The diameter of a $H$-minor free graph can be computed in time $\tilde{\cal O}(n^{2-\varepsilon_H})$, with a Monte Carlo algorithm, where $\varepsilon_H \in (0;1)$ is a constant that only depends on $H$.
\end{corollary}

For most values of $H$ this is the first known subquadratic-time algorithm for diameter computation on $H$-minor free graphs.
In particular, this is the first known subquadratic algorithm for diameter computation on (unweighted) bounded-genus graphs to the best of our knowledge (see the planar graphs paragraph in the introduction).

\subsection{Proof of Theorem~\ref{thm:subLin}}\label{sec:proof-sublin}

The remainder of this section is devoted to the proof of Theorem~\ref{thm:subLin}.
We start by presenting, in a separate subsection, all the required background on $r$-divisions.

\subsubsection*{Algorithmic aspects of $r$-divisions}

Throughout all this section, let ${\cal G}_{\alpha,C}$ be the class of all the graphs $G$ such that, for every connected $h$-vertex subgraph of $G$, there exists a balanced separator of order at most $C \cdot h^{\alpha}$. 
The following intermediate result is built upon a previous algorithm from Plotkin et al.~\cite{PRS94}.

\begin{lemma}[\cite{Dvo18}]\label{lem:comput-sublin}
For every $n$-vertex $m$-edge graph $G \in {\cal G}_{\alpha,C}$, we can find a balanced separator of order ${\cal O}(n^{\frac{4+\alpha}{5}})$ in time ${\cal O}(mn^{\frac{4+\alpha}{5}}) = {\cal O}(n^{2-\frac{1-\alpha}5})$.
\end{lemma}

We will also use the following simple result:
\begin{lemma}\label{lem:balanced-bipartition}
Let $G$ be a graph and $S$ a balanced separator.
We can bipartition the connected components of $G \setminus S$ in two disjoint sets $A$ and $B$ such that $\min\{ |A|, |B| \} \leq 2n/3$.
\end{lemma}
\begin{proof}
Let $C_1,C_2,\ldots,C_k$ be the connected components of $G \setminus S$.
We define $i_0 := \max\{ i \mid |\bigcup_{j < i} C_j| \leq 2n/3 \}$.  
Let $A' := \bigcup_{j < i_0} C_j$ and $B' := \bigcup_{j > i_0} C_j$.
If $|B' \cup C_{i_0}| \leq 2n/3$ then we are done by setting $A := A', \ B := B' \cup C_{i_0}$.
Thus, from now on let us assume that $|B' \cup C_{i_0}| > 2n/3$.
Note that since $S$ is a balanced separator, it implies that $i_0 < k$.
Then, by the very definition of $i_0$ we also have $|A' \cup C_{i_0}| > 2n/3$.
Overall, $|A'| + 2|C_{i_0}| + |B'| >4n/3$.
Since $|A'| + |B'| + |C_{i_0}| < n$, we obtain $|C_{i_0}| > n/3$.
We are done by setting $A := A' \cup B'$ and $B := C_{i_0}$.
\end{proof}

Now, set $\beta := \frac{4+\alpha}{5} < 1$\footnote{More generally, let ${\cal G} \subseteq {\cal G}_{\alpha,C}$. We may choose any parameter $\beta \in [\alpha;1)$ such that for all the graphs in ${\cal G}$ we can compute a balanced separator of size ${\cal O}(n^{\beta})$ in truly subquadratic-time. For instance by Lemma~\ref{lem:h-minor}, if ${\cal G}$ is proper minor-closed then we can set $\beta = \alpha = 1/2$.}.
By Lemma~\ref{lem:comput-sublin}, for every $n$-vertex $m$-edge graph in ${\cal G}_{\alpha,C}$ we can compute a balanced separator of order ${\cal O}(n^{\beta})$ in time ${\cal O}(n^{1+\beta})$.
Following Federickson~\cite{Fed87}, we define an {\em $r$-division} for an $n$-vertex graph $G \in {\cal G}_{\alpha,C}$ as follows:
\begin{itemize}
\item If $n \leq r$ then, we output $G$;
\item Otherwise, let $S$ be a balanced separator of cardinality ${\cal O}(n^{\beta})$. Since $S$ is balanced then, by Lemma~\ref{lem:balanced-bipartition} we can partition the connected components of $G \setminus S$ in two disjoint sets $A$ and $B$ of cardinality $\leq 2n/3$. We end up computing an $r$-division for the induced subgraphs $G[A \cup S]$ and $G[B \cup S]$ separately. Note that since $S$ is a separator, all edges of $G$ are covered by these two subgraphs.
\end{itemize}
Therefore by construction, an $r$-division of a connected graph $G$ is a collection of connected induced subgraphs of order at most $r$ that cover all edges of $G$.
We will use the terminology from~\cite{HaQ17}.
In particular, the subgraphs in an $r$-division are termed {\em clusters}.
A vertex is {\em interior} if it is contained in a unique cluster, otherwise it is a {\em boundary} vertex.
Finally, if the sum of the orders of all the clusters is $n + q$ then, we call $q$ the {\em excess}. 

The following result is essentially a reformulation of~\cite[Lemma 2.2]{HaQ17}.

\begin{lemma}[\cite{HaQ17}]\label{lem:sum-order}
Set $\beta := \frac{4+\alpha}{5}$.
There exists a constant $r_0$ such that, for any $n$-vertex graph $G \in {\cal G}_{\alpha,C}$ and $r \geq r_0$, any $r$-division of $G$ has an excess in ${\cal O}(n/r^{1-\beta})$.
\end{lemma}

Note that in our applications, we will choose $r = n^{\gamma}$ for some $\gamma \in (0;1)$ that only depends on $\beta$ and on the distance VC-dimension.

\smallskip
It is easy to prove that an $r$-division can be computed in polynomial time~\cite{HaQ17}.
Next we use the known connections between strongly sublinear separators and {\em polynomial expansion}~\cite{Dvo18} in order to bound the running-time by some truly subquadratic function.

\begin{corollary}\label{cor:rdivision-runtime}
Set $\beta := \frac{4+\alpha}{5}$.
Then, for any $n$-vertex $m$-edge graph $G \in {\cal G}_{\alpha,C}$, we can compute an $r$-division in time $\tilde{\cal O}(n^{1+\beta})$.
\end{corollary}

\begin{proof}
Let us assume that at the initialization step, $n > r$ (otherwise, we are done). 
We claim that it is sufficient to prove that the total number of edges in the final clusters is in ${\cal O}(n)$. 
Indeed, if this is true for the final clusters then, this is also true for the intermediate clusters at any given step of the decomposition.
In particular, every step runs in time ${\cal O}(n^{1+\beta})$.
Furthermore, since we only consider balanced separators of {\em sublinear} cardinality, for every $n$ above some constant the two induced subgraphs constructed have truly sublinear order (say, $\leq 3n/4$).
Therefore it takes ${\cal O}(\log{n})$ steps to decrease the order of all the subgraphs in this collection to less than $r$.
This upper-bound on the number of steps proves, as claimed, that the total running time is in $\tilde{\cal O}(n^{1+\beta})$. 

\smallskip
We are left proving that the total number of edges in the final clusters is indeed in ${\cal O}(n)$.
For that, let us consider any of the clusters $W_i$.
Since ${\cal G}_{\alpha,C}$ is monotone, we have $W_i \in {\cal G}_{\alpha,C}$.
Furthermore, every graph in ${\cal G}_{\alpha,C}$ must be ${\cal O}(1)$-degenerate ({\it e.g.}, see~\cite[Lemma 2 (b)]{Dvo18} where the author proved a stronger result, namely that ${\cal G}_{\alpha,C}$ has polynomial expansion).
It implies that $W_i$ has size ${\cal O}(|V(W_i)|)$.
Overall, if the total excess is $q$ then, the total number of edges in the clusters is in ${\cal O}(n+q)$.
By Lemma~\ref{lem:sum-order} we have $q = {\cal O}(n)$, and so the total number of edges is also in ${\cal O}(n)$.
\end{proof} 

\subsubsection*{Boundary Hypergraphs}

Let $G$ be a graph equipped with some $r$-division, and let $\ell$ be a positive integer.
Roughly, our objective is to use the $r$-division in order to compute, for every vertex, a compact interval representation of its balls of radius $\ell$.
This leads us to the following natural object:

\begin{definition}\label{def:boundary-hypergraph}
Let $\Lambda_r$ be an $r$-division of a graph $G$, and let $\ell$ be a positive integer.
The {\em $\ell$-boundary hypergraph} ${\cal H}_{\ell,G}(\Lambda_r)$ has for vertex-set $V$.
Moreover, for every cluster $W_i \in \Lambda_r$ and for every $u,v \in V(W_i)$, if $v$ is a boundary vertex and $dist_G(u,v) < \ell$, then the ball $N_G^{\ell-dist_G(u,v)}[v]$ is a hyperedge of ${\cal H}_{\ell,G}(\Lambda_r)$.
\end{definition}

To better understand this above construction, let $W_i$ be a cluster, let $u \in V(W_i)$ be internal and let $z \notin V(W_i)$.
Then, since an $r$-division is also an edge-covering, we have $dist_G(u,z) \leq \ell$ if and only if there exists a boundary vertex $v \in V(W_i)$ such that $dist_G(u,v) + dist_G(v,z) \leq \ell$.
Equivalently, we have $z \in N_G^{\ell-dist_G(u,v)}[v]$.

\begin{lemma}\label{lem:boundary-num}
Set $\beta := \frac{4+\alpha}{5}$.
Then, for any $n$-vertex graph $G \in {\cal G}_{\alpha,C}$, and for any $r$-division $\Lambda_r$, the $\ell$-boundary hypergraph ${\cal H}_{\ell,G}(\Lambda_r)$ has ${\cal O}(nr^{\beta})$ hyperedges.
\end{lemma}

\begin{proof}
For every $W_i \in \Lambda_r$, we create ${\cal O}(r \cdot b_i)$ hyperedges, where $b_i$ denotes the number of boundary vertices in the cluster.
We observe that $\sum_{W_i \in \Lambda_r} b_i$ is at most twice the excess.
Then, by Lemma~\ref{lem:sum-order} we have ${\cal O}(r) \times {\cal O}(n/r^{1-\beta}) = {\cal O}(nr^{\beta})$ hyperedges.
\end{proof}

We stress that by Lemma~\ref{lem:boundary-num}, a boundary hypergraph may have a superlinear number of edges.
Therefore, if we restrict ourselves to subquadratic-time computation, we cannot compute this hypergraph explicitly.
Fortunately, we show next that this is not needed if one just wants to compute for this hypergraph a spanning path of low stabbing number.

\begin{lemma}\label{lem:boundary-spanning-path}
Set $\beta := \frac{4+\alpha}{5}$, and let $G \in {\cal G}_{\alpha,C}$ have distance VC-dimension at most $d$.
Then, there exists a constant $\varepsilon_d \in (0;1)$ that only depends on $d$ and such that, for any $r$-division $\Lambda_r$, the stabbing number of ${\cal H}_{\ell,G}(\Lambda_r)$ is in $\tilde{\cal O}(n^{1-\varepsilon_d})$.
Moreover, we can compute a spanning path reaching this upper bound in deterministic time $\tilde{\cal O}(n^2/r^{1-\beta} + n^{2-\varepsilon_d}r^{\beta})$.  
\end{lemma}

\begin{proof}
By construction, ${\cal H}_{\ell,G}(\Lambda_r)$ is a subhypergraph of ${\cal B}(G)$, the ball hypergraph of $G$.
Therefore, the VC-dimension of ${\cal H}_{\ell,G}(\Lambda_r)$ is at most $d$.
Let $\varepsilon_d$ be the constant of Theorem~\ref{thm:stabbingNb}.
In order to prove the result, we are left proving that we can adapt the algorithm of Theorem~\ref{thm:stabbingNb} so that it runs in time $\tilde{\cal O}(nm/r^{1-\beta} + n^{2-\varepsilon_d}r^{\beta})$ when it is given ${\cal H}_{\ell,G}(\Lambda_r)$ as input.
For that, let $F$ be the set of the boundary vertices.
We have that $|F|$ is at most twice the excess, and so, by Lemma~\ref{lem:sum-order} we get $|F| = {\cal O}(n/r^{1-\beta})$.
\begin{enumerate}
\item
We start with a breadth-first search from every vertex of $F$.
This pre-processing phase takes time ${\cal O}(|F|m) = {\cal O}(n^2/r^{1-\beta})$.
Furthermore, note that in doing so we computed all the pairs $(v,t) \in F \times [\ell]$ such that $N_G^t[v]$ is a hyperedge of ${\cal H}_{\ell,G}(\Lambda_r)$.   

\item Let $\eta = 2^{d+1}\varepsilon_d$.
We arbitrarily partition the vertex-set $V$ into subsets $V_1,V_2,\ldots,V_p$ such that $p = {\cal O}(n^{1-\eta})$ and, for every $1 \leq i \leq p$, $|V_i| = {\cal O}(n^{\eta})$.
Furthermore, as explained in the proof of Theorem~\ref{thm:stabbingNb} (i.e., Claim~\ref{claim:stabbing-partition}), we can compute a spanning path of stabbing number $\tilde{\cal O}(n^{1-\varepsilon_d})$ for ${\cal H}_{\ell,G}(\Lambda_r)$ if we are given the subhypergraphs ${\cal H}_1, {\cal H}_2, \ldots, {\cal H}_p$ that are induced by $V_1,V_2,\ldots, V_p$ respectively. 
It takes time $\tilde{\cal O}(n^{1 + \eta[c(d+1)-1]})$ for some constant $c$, that is in $\tilde{\cal O}(n^{2-\varepsilon_d})$.

\medskip
In order to compute all the subhypergraphs ${\cal H}_i$, we could proceed by brute-force, as follows.
For every $i$ and for any boundary vertex $v$, we read the vertices of $V_i$ by non-decreasing distance to $v$. 
Furthermore, if $N_G^t[v]$ is a hyperedge of ${\cal H}_{\ell,G}(\Lambda_r)$, then as soon as we exceed distance $t$ all the vertices read so far are exactly $N_G^t[v] \cap V_i$.
Overall, for a fixed boundary vertex $v$ we could obtain this way up to ${\cal O}(|V_i|)$ different subsets of order ${\cal O}(|V_i|)$ each.
But unfortunately, that would give us a time complexity in ${\cal O}(|F||V_i|^2) = {\cal O}(n^{1+2\eta}/r^{1-\beta})$ for a given $i$, and so a total running time in ${\cal O}(n^{2+\eta}/r^{1-\beta})$.
In order to lower this running-time, we proceed as follows.

\begin{enumerate}
\item For every $v \in F$, we group all the vertices in $V_i$ at equal distance to $v$.
We totally order this partition by increasing distance of its vertices to $v$.
Doing so we get exactly $n_i := |V_i|$ ordered groups (possibly, by adding some empty groups in the sequence), denoted $V_i^1(v), V_i^2(v), \ldots, V_i^{n_i}(v)$. 
Overall, this phase takes time ${\cal O}(|F||V_i|) = {\cal O}(n^{1+\eta}/r^{1-\beta})$.

\item Then, we introduce a complex subprocedure in order to gradually remove the duplicates from the sets $N_G^t[v] \cap V_i, \ \mbox{for} \ v \in F \ \mbox{and} \ t \geq 0$. For every $j = 0\ldots n_i$, we map every boundary vertex $v$ to $\bigcup_{j' \leq j} V_i^{j'}(v)$. More precisely, we maintain some collection of different subsets of $V_i$, denoted ${\cal P}_j = \left( V_i^{j,1}, V_i^{j,2}, \ldots, V_i^{j,s_i(j)} \right)$ (note that ${\cal P}_j$ is a list of lists). For every $v \in F$ we ensure that there is a unique $t$ such that $V_i^{j,t} =    \bigcup_{j' \leq j} V_i^{j'}(v)$.
Then, there is a pointer from vertex $v$ to this $t^{th}$ subset (equivalently, for every list in ${\cal P}_j$, we store an auxiliary list of all the corresponding vertices of $F$).

\smallskip
We will show next that it is easy to construct ${\cal P}_{j+1}$ from ${\cal P}_j$, but that the natural method for doing so might generate some duplicates. Roughly, by using in our analysis the Sauer-Shelah-Perles lemma, we prove that it is more efficient to remove duplicates at every single step rather than doing it only once at the end of the subprocedure. 

\smallskip
We observe that initially for $j=0$, there is a unique subset $V_i^{0,1} = \emptyset$. 
Furthermore if all the subsets $V_i^{j,t}$ have been computed at step $j$, then we can compute those at step $j+1$, as follows: 
\begin{itemize}
\item For every $v \in F$, if we have $V_i^{j,t} = \bigcup_{j' \leq j} V_i^{j'}(v)$, then we add a copy of $V_i^{j+1}(v)$ into some buffer $b_{j+1}'(t)$ and a pointer from $v$ to this copy. It takes time ${\cal O}(\sum_{v \in F} |V_i^{j+1}(v)|)$. 
\item Then, for every $1 \leq t \leq s_i(j)$, we remove all the duplicated subsets in the buffer $b_{j+1}'(t)$.
The new buffer that we get is denoted $b_{j+1}(t)$.
We can compute it by using partition refinement ({\it e.g.}, see~\cite{HMPV00}), that takes time ${\cal O}(\sum_{W \in b_{j+1}'(t)}|W|)$ up to some ${\cal O}(|V_i|)$-time pre-processing.
Overall the removal of all the duplicates, for all $t$, takes total time ${\cal O}(n^{\eta} + \sum_{v \in F} |V_i^{j+1}(v)|)$.
Furthermore on our way to remove the duplicates, we also need to actualize the pointers between the boundary vertices and the buffer contents, that takes additional time ${\cal O}(|F|) = {\cal O}(n/r^{1-\beta})$.
\item For every $1 \leq t \leq s_i(j)$, we can now refine $V^{j,t}_i$ in $|b_{j+1}(t)|$ new subsets.
Every such subset has order ${\cal O}(n^{\eta})$, and so this operation takes total time ${\cal O}(n^{\eta}|b_{j+1}(t)|)$.
Overall, we obtain a new collection of ${\cal O}(\sum_t |b_{j+1}(t)|)$ subsets. 
Furthermore, on our way to construct this collection, we can add a pointer from every boundary vertex $v$ to {\em one} subset equal to $\bigcup_{j' \leq j+1}V_i^{j'}(v)$ (there may be duplicated subsets). 
By carefully using the pointers added between the boundary vertices and the buffer contents during the previous phases, this operation takes additional time ${\cal O}(|F|) = {\cal O}(n/r^{1-\beta})$.

\item Finally, since all the subsets in the new collection have order ${\cal O}(n^{\eta})$, by using again partition refinement we can merge all the duplicated subsets in time ${\cal O}(|V_i| + n^{\eta} \cdot \sum_t |b_{j+1}(t)|) = {\cal O}(n^{\eta} \cdot \sum_t |b_{j+1}(t)|)$. 
We also need to actualize the pointers between the boundary vertices and the subsets, that takes total time ${\cal O}(|F|) = {\cal O}(n/r^{1-\beta})$. 
\end{itemize}
Let us upper bound $s_i(j)$. 
For that we stress that every subset $V^{j,t}_i$ represents a different intersection of $V_i$ with a ball of $G$, hence of a hyperedge of ${\cal B}(G)$.
Since ${\cal B}(G)$ has VC-dimension at most $d$, by Lemma~\ref{lem:vc-sub} so does its subhypergraph ${\cal H}_i'$ induced by $V_i$.
In particular, every $V^{j,t}_i$ is a hyperedge of ${\cal H}_i'$.
By Lemma~\ref{lem:sauer-lemma} we get that $s_i(j) = {\cal O}(n^{\eta d})$.
In the same way, since for a fixed $t$ the $|b_{j+1}(t)|$ new subsets that are obtained by refinement of $V_i^{j,t}$ are pairwise different, we have $|b_{j+1}(t)| \leq s_i(j+1) = {\cal O}(n^{\eta d})$. 
As a result, the passing from step $j$ to step $j+1$ takes time: $${\cal O}\left(\left[\sum_{v \in F} \left|V_i^{j+1}(t)\right|\right] + n/r^{1-\beta} + n^{\eta} \cdot n^{\eta d} \cdot n^{\eta d}\right) = {\cal O}\left(\left[\sum_{v \in F} \left|V_i^{j+1}(t)\right|\right] + n/r^{1-\beta} + n^{(2d+1)\eta}\right).$$
There are ${\cal O}(n^{\eta})$ loops, that gives us a total running time of: 
\begin{align*}
{\cal O}\left(\left[\sum_{v \in F}\sum_{j=1}^{n_i} \left|V_i^{j+1}(t)\right|\right] + n^{1+\eta}/r^{1-\beta} + n^{(2d+2)\eta}\right) &= {\cal O}\left(\left[\sum_{v \in F}n^{\eta}\right] + n^{1+\eta}/r^{1-\beta} + n^{2(d+1)\eta}\right) \\
&= {\cal O}\left(n^{1+\eta}/r^{1-\beta} + n^{2(d+1)\eta}\right).
\end{align*}
\item Here the key observation is that $\bigcup_j {\cal P}_j$ contains the intersection with $V_i$ of {\em all} the balls whose center is in $F$. We so computed a {\em superset} of order ${\cal O}(n^{(d+1)\eta})$ (i.e., ${\cal O}(n^{\eta d})$ per loop) that contains all possible intersections between a hyperedge of ${\cal H}_{\ell,G}(\Lambda_r)$ and $V_i$. Since every subset in $\bigcup_j {\cal P}_j$ represents the intersection of a hyperedge of ${\cal B}(G)$ with $V_i$, and furthermore ${\cal B}(G)$ has VC-dimension at most $d$, then for simplicity we may replace ${\cal H}_i$ by the slightly larger hypergraph ${\cal H}_i'$ of which these are the hyperedges (i.e., the hyperedges of ${\cal H}_i'$ are the intersections of $V_i$ with all the balls whose center is in $F$).
Note that in order to compute ${\cal H}_i'$, it is sufficient to eliminate all the duplicated elements in this collection $\bigcup_j {\cal P}_j$, that takes total time ${\cal O}(n^{(d+2)\eta})$.
\end{enumerate}
The running-time is in $\tilde{\cal O}(n^{1+\eta}/r^{1-\beta} + n^{2(d+1)\eta})$ for any fixed $i$.
Therefore, the total running-time is in $\tilde{\cal O}(n^{2}/r^{1-\beta} + n^{1 + [2(d+1)-1]\eta})$.
Recall (see Theorem~\ref{thm:stabbingNb} and its proof) that we have $\tilde{\cal O}(n^{1 + \eta[c(d+1)-1]}) = \tilde{\cal O}(n^{2-\varepsilon_d})$ for some constant $c > 2$.
As a result, the running-time of this part is also in $\tilde{\cal O}(n^{2}/r^{1-\beta} + n^{2-\varepsilon_d})$.
\item
By continuing the algorithm of Theorem~\ref{thm:stabbingNb} with the hypergraphs ${\cal H}_1', {\cal H}_2', \ldots, {\cal H}_p'$, we get a spanning path of ${\cal H}_{\ell,G}(\Lambda_r)$ whose stabbing number is in $\tilde{\cal O}(n^{1-\varepsilon_d})$.
It remains to compute, for every hyperedge of ${\cal H}_{\ell,G}(\Lambda_r)$, the set of the stabbed edges.
For that, let $v \in F$ be fixed.
We add all the radii $t$ such that $N_G^t[v]$ is a hyperedge of ${\cal H}_{\ell,G}(\Lambda_r)$ in a balanced binary research tree $T_v$.
Then, we scan all the edges $xy$ of the spanning path.
By symmetry let us assume that $dist_G(v,x) \leq dist_G(v,y)$.
The edge $xy$ is stabbed by all the hyperedges $N_G^t[v]$ such that $dist_G(v,x) \leq t < dist_G(v,y)$.
Then by using $T_v$, after some pre-computation in time ${\cal O}(\log{n})$ every value $t$ in the range $[dist_G(v,x);dist_G(v,y))$ can be enumerated in constant-time.

Overall, by Lemma~\ref{lem:boundary-num} there are ${\cal O}(nr^{\beta})$ edges, and so the construction of all the balanced binary research trees takes time $\tilde{\cal O}(nr^{\beta})$.
Scanning all the edges, for every boundary vertex, takes total time $\tilde{\cal O}(n^2/r^{1-\beta})$.
Any other operation corresponds to an edge of the spanning path that is stabbed by a hyperedge of ${\cal H}_{\ell,G}(\Lambda_r)$, and as a result there can only be $\tilde{\cal O}(n^{1-\varepsilon_d}) \times {\cal O}(nr^{\beta}) = \tilde{\cal O}(n^{2-\varepsilon_d}r^{\beta})$ such operations.
\end{enumerate}
Altogether combined, the running time of the algorithm is in $\tilde{\cal O}(n^2/r^{1-\beta} + n^{2-\varepsilon_d}r^{\beta})$.  
\end{proof}

\subsubsection*{The algorithm}

We are now ready to prove the main result of this section.

\begin{proofof}{Theorem~\ref{thm:subLin}}
By a classical dichotomic argument it is sufficient to prove that for any $k$, we can decide whether $diam(G) \leq k$ in truly subquadratic time.
Furthermore in order to solve this decision problem, by Lemma~\ref{lem:reduction} we are left with computing a spanning path of strongly sublinear stabbing number for the $(k-1)$-neighbourhood hypergraph.
For that, let $b_0b_1\ldots b_{s-1}b_s$ be the binary decomposition of $k-1$ (from the most significant to the less significant bit).
Furthermore, for every $0 \leq i \leq s$, let $k_i$ be the positive integer of binary decomposition $b_0b_1\ldots b_{i-1}b_i$.
We will prove by finite induction that one can compute a spanning path of strongly sublinear stabbing number for the $k_i$-neighbourhood hypergraph of $G$.
More precisely, let $\varepsilon_d$ be the constant of Theorem~\ref{thm:stabbingNb}.
We will compute such spanning path of stabbing number $\tilde{\cal O}(n^{1-\varepsilon_d})$ in $\tilde{\cal O}(n^{2-\varepsilon'})$ time for some $\varepsilon'>0$. 
Note that since $s = {\cal O}(\log{k})$, that will indeed give us a truly subquadratic algorithm for deciding whether $diam(G) \leq k$.

If $i=0$, then $b_0 = 1$ and the result follows from Theorems~\ref{thm:diamTwo} and~\ref{thm:stabbingNb}.
Thus from now on assume $i > 0$.
We observe that $k_i = 2k_{i-1} + b_i$.
Furthermore if $b_i = 1$, then we explained in the proof of Theorem~\ref{thm:boundedDiam} (Subclaims~\ref{subclaim:subopt} and~\ref{subclaim:apply-eps-net}) how to compute a spanning path of stabbing number $\tilde{\cal O}(n^{1-\varepsilon_d})$ for ${\cal N}_{k_i}(G)$ from such spanning path for ${\cal N}_{k_i-1}(G) = {\cal N}_{2k_{i-1}}(G)$; it takes $\tilde{\cal O}(mn^{1-\varepsilon_d}) = \tilde{\cal O}(n^{2-\varepsilon_d})$ time, that is truly subquadratic. 
In order to complete the proof, we now adapt this algorithm of Theorem~\ref{thm:boundedDiam} so as to compute a spanning path with stabbing number $\tilde{\cal O}(n^{1-\varepsilon_d})$ for ${\cal N}_{2k_{i-1}}(G)$ from such spanning path for ${\cal N}_{k_{i-1}}(G)$. 
Specifically, let $C$ and $\alpha < 1$ be such that ${\cal G} \subseteq {\cal G}_{\alpha,C}$ and set $\beta := \frac{4+\alpha}5 < 1$. The running time for computing the adequate spanning path will be $\tilde{\cal O}\left(n^{2-\left(1- \frac 1 {2-\beta}\right)\varepsilon_d}\right)$.
We first prove the following intermediate result for any value $r>0$ (we will use $r=\Theta(n^{\frac{\varepsilon_d}{2-\beta}})$):

\begin{myclaim}\label{claim:subopt-sublin}
Let $P_{i-1}$ be a spanning path of stabbing number $\tilde{\cal O}(n^{1-\varepsilon_d})$ for ${\cal N}_{k_{i-1}}(G)$.
In $\tilde{\cal O}(n^{2-\varepsilon_d}r + n^2/r^{1-\beta})$ time, we can compute from $P_{i-1}$ a spanning path $P_i'$ for ${\cal N}_{2k_{i-1}}(G)$, such that \\ $\sum_{v \in V} |E_{P_i'}(N_G^{2k_{i-1}}[v])| = \tilde{\cal O}(n(r + n^{1-\varepsilon_d} r^{\beta}))$.
\end{myclaim}

\begin{proofclaim}
By Lemma~\ref{cor:rdivision-runtime} we can compute an $r$-division, denoted $\Lambda_r$, in time $\tilde{\cal O}(n^{1+\beta}) = \tilde{\cal O}(n^2/r^{1-\beta})$.
Then, we proceed as follows.
\begin{enumerate}
\item
We first consider all the clusters $W \in \Lambda_r$ sequentially. For every $x,y \in W$ we have $dist_G(x,y) \leq 2k_{i-1}$ if and only if we have $N_G^{k_{i-1}}[x] \cap N_G^{k_{i-1}}[y] \neq \emptyset$.  
Given a spanning path of stabbing number $\tilde{\cal O}(n^{1-\varepsilon_d})$ for ${\cal N}_{k_{i-1}}(G)$, this test can be easily done in time $\tilde{\cal O}(n^{1-\varepsilon_d})$ (i.e., by sorting the ends of the $\tilde{\cal O}(n^{1-\varepsilon_d})$ intervals on this spanning path that cover the balls $N_G^{k_{i-1}}[x]$ and $N_G^{k_{i-1}}[y]$).
Furthermore by Lemma~\ref{lem:sum-order} we have $\sum_{W \in \Lambda_r} |W| = \Theta(n)$, and so this step takes time $\tilde{\cal O}(n^{1-\varepsilon_d}) \times \sum_{W \in \Lambda_r} {\cal O}(|W|^2)= \tilde{\cal O}(rn^{1-\varepsilon_d}) \times \sum_{W \in \Lambda_r} |W| = \tilde{\cal O}(rn^{2-\varepsilon_d})$.
Overall for every $u \in V$, we computed all the vertices of $N_G^{2k_{i-1}}[u]$ that belong to a common cluster with $u$.
\item
Let us now consider the $2k_{i-1}$-boundary hypergraph ${\cal H}_{2k_{i-1},G}(\Lambda_r)$. By Lemma~\ref{lem:boundary-spanning-path} we can compute a spanning path $P_i'$ of stabbing number $\tilde{\cal O}(n^{1-\varepsilon_d})$ for this hypergraph, in time $\tilde{\cal O}(nm/r^{1-\beta} + n^{2-\varepsilon_d}r^{\beta}) = \tilde{\cal O}(n^2/r^{1-\beta} + n^{2-\varepsilon_d}r^{\beta})$.
Let $u \in V$.
There are two cases:
\begin{itemize}
\item{\it Case $u$ is a boundary vertex.} Since $N_G^{2k_{i-1}}[u]$ is a hyperedge of the boundary hypergraph, we have $|E_{P_{i}'}(N_G^{2k_{i-1}}[u])| = \tilde{\cal O}(n^{1-\varepsilon_d})$ (already computed).
\item{\it Case $u$ is an internal vertex.}
Let $W \in \Lambda_r$ be the unique cluster containing $u$, and set initially $E_{P_i'}(N_G^{2k_{i-1}}[u]) := \emptyset$. 
For every boundary vertex $v \in V(W)$, if $dist_G(u,v) < 2k_{i-1}$ then, we add all of $E_{P_i'}(N_G^{2k_{i-1}-dist_G(u,v)}[v])$ to $E_{P_i'}(N_G^{2k_{i-1}}[u])$.
Assuming there are $b_W$ boundary vertices in $W$, we obtain that $|E_{P_i'}(N_G^{2k_{i-1}}[u])| = \tilde{\cal O}(b_W \cdot n^{1-\varepsilon_d})$.
Furthermore, this above set of stabbed edges defines a collection of intervals that covers exactly the balls $N_G^{2k_{i-1}-dist_G(u,v)}[v]$, for the boundary vertices $v \in V(W)$.
Denote this collection of intervals by $I_i(u)$, and its complementary by $\overline{I_i(u)}$.
By construction, every vertex on an interval of $I_i(u)$ is at a distance $\leq 2k_{i-1}$ to $u$; conversely, since $\Lambda_r$ is also an edge-covering, every vertex of $N_G^{2k_{i-1}}[u] \setminus W$ must be on one of these intervals.
As a result, in order to construct $E_{P_i'}(N_G^{2k_{i-1}}[u])$, it suffices to update this set using the vertices of $N_G^{2k_{i-1}}[u] \cap W$ that lie on some interval of $\overline{I_i(u)}$.
Note that in doing so, we can only modify the cardinality of $E_{P_i'}(N_G^{2k_{i-1}}[u])$ by an ${\cal O}(|W|) = {\cal O}(r)$. 
\end{itemize}
Overall, we obtain that $\sum_{u \in V} |E_{P_i'}(N_G^{2k_{i-1}}[u])| = \tilde{\cal O}(nr + n^{1-\varepsilon_d} \cdot \sum_W (b_W \cdot |V(W)|)) = \tilde{\cal O}(nr + rn^{1-\varepsilon_d} \cdot \sum_W b_W)$. Again we observe that $\sum_W b_W$ is at most twice the excess, and so by Lemma~\ref{lem:sum-order} $\sum_W b_W = {\cal O}(n/r^{1-\beta})$.
Therefore, $\sum_{u \in V} |E_{P_i'}(N_G^{2k_{i-1}}[u])| = \tilde{\cal O}(n(r + n^{1-\varepsilon_d} r^{\beta}))$.
\end{enumerate}
\end{proofclaim}
By combining Claim~\ref{claim:subopt-sublin} with Subclaim~\ref{subclaim:apply-eps-net}, in $\tilde{\cal O}(n^{2-\varepsilon_d}r + n^2/r^{1-\beta} + n^{1-\varepsilon_d}m)$ time we compute a spanning path of stabbing number $\tilde{\cal O}(n^{1-\varepsilon_d})$ for ${\cal N}_{2k_{i-1}}(G)$, thereby completing the description of our algorithm.
Overall, the running-time of our algorithm is optimized when we have $n^2/r^{1-\beta} = n^{2-\varepsilon_d}r$.
As a result, a good choice is $r = \Theta(n^{\frac{\varepsilon_d}{2-\beta}})$.
Finally, we stress that in this case, $n^{2-\varepsilon_d}r = n^{2-\left(1- \frac 1 {2-\beta}\right)\varepsilon_d}$ is truly subquadratic, because $\beta < 1$ and so $1- \frac 1 {2-\beta} > 0$.
\end{proofof}

\section{Open problems}\label{sec:ccl}
We left open whether we can compute the diameter of all the graphs of constant distance VC-dimension and unbounded diameter in truly subquadratic time.
In order to prove that it is the case, we stress that by our Theorem~\ref{thm:boundedDiam} we only need to consider the graphs of large diameter, i.e., above some polynomial.

Another interesting direction could be to extend the constant distance VC-dimension framework to weighted graphs.

\smallskip
Finally, we observe that there exist graph families of {\em unbounded} (distance) VC-dimension for which we can compute the diameter very efficiently. 
For instance, the class of all the graphs with a universal vertex has unbounded VC-dimension.
Such graphs are a particular case of {\em dually chordal} graphs, for which we know how to compute the diameter in linear time~\cite{BCDV98}.
We observe that the ball hypergraphs of dually chordal graphs also admit some nice characterizations.
Thus, it would be very interesting to study whether a truly subquadratic algorithm for computing the diameter could be derived from some common property of dually chordal graphs and graphs of constant distance VC-dimension (say, a bounded fractional Helly number~\cite{Mat04}).


\end{document}